\documentclass[journal]{IEEEtran}
\IEEEoverridecommandlockouts
\usepackage{cite}
\usepackage{amsmath,amssymb,amsfonts}
\usepackage{algorithmic}
\usepackage{graphicx}
\usepackage{textcomp}
\usepackage{xcolor}
\usepackage{amsmath}
\usepackage{amsthm}
\usepackage{bm}
\usepackage[export]{adjustbox}
\usepackage{threeparttable}
\usepackage{subfigure}
\usepackage{booktabs}

\newtheorem{theorem}{Theorem}
\newtheorem{lemma}{Lemma}
\def\BibTeX{{\rm B\kern-.05em{\sc i\kern-.025em b}\kern-.08em
    T\kern-.1667em\lower.7ex\hbox{E}\kern-.125emX}}
\begin{document}

\title{From Spectrum Wavelet to Vertex Propagation: Graph Convolutional Networks
	Based on Taylor Approximation
%
}
\author{Songyang~Zhang,~\IEEEmembership{Member,~IEEE,}
	Han~Zhang, Shuguang~Cui,~\IEEEmembership{Fellow,~IEEE,}~%
	and~Zhi~Ding,~\IEEEmembership{Fellow,~IEEE}
\IEEEcompsocitemizethanks{\IEEEcompsocthanksitem 
	S. Zhang, and Z. Ding are with Department of Electrical and Computer Engineering, Univerisity of California, Davis, CA, 95616.\protect\\ (E-mail: sydzhang@ucdavis.edu, and zding@ucdavis.edu).
\IEEEcompsocthanksitem 
H. Zhang was with Department of Electrical and Computer Engineering, Univerisity of California, Davis, CA, 95616.\protect\\ (E-mail: hanzh@ucdavis.edu).
\IEEEcompsocthanksitem S. Cui is currently with the Shenzhen Research Institute of Big Data and Future Network of Intelligence Institute (FNii), the Chinese University of Hong Kong, Shenzhen, China, 518172. \protect\\(E-mail: shuguangcui@cuhk.edu.cn).}}%

\maketitle
\begin{abstract}
	Graph convolutional networks (GCN) have been recently utilized to 
	extract the underlying structures of datasets with some labeled data and high-dimensional features.
	Existing GCNs mostly rely on a first-order Chebyshev approximation 
	of graph wavelet-kernels. 
Such a generic propagation model does not
always suit the various datasets and their features.
	This work revisits the fundamentals of graph wavelet and explores
	the utility of
	signal propagation in the vertex domain to approximate the spectral wavelet-kernels.
	We first derive the conditions for representing the graph wavelet-kernels 
	via vertex propagation. 
	We next propose alternative propagation models for GCN layers based on Taylor expansions. 
	We further analyze the choices of detailed graph representations for TGCNs.
	Experiments on citation networks, multimedia datasets and synthetic graphs demonstrate the advantage of Taylor-based GCN (TGCN) in the node classification problems over the traditional GCN methods.
	
\end{abstract}


\section{Introduction}
\IEEEPARstart{L}{earning} and signal processing over graph models have gained significant
traction owing to their demonstrated abilities 
to capture underlying 
data interactions.
Modeling each data point as a node and their interactions as edges in a graph, 
graph-based methods have been adopted in various signal
processing and analysis tasks, such as  
semi-supervised classification \cite{kipf2017semi,zhu2003semi}, spectral clustering \cite{white2005spectral,buhler2009spectral}, 
link prediction \cite{grover2016node2vec} and graph classification \cite{niepert2016learning}. 
For example, in \cite{8101025}, an undirected graph is constructed based on a Gaussian-distance model to 
capture geometric correlations among points in a 
point cloud, with which several graph-based filters 
have been developed to extract contour features 
of objects.

Among various graph-based tools, graph signal processing  
(GSP) has emerged as an efficient analytical tool for processing graph-modeled signals \cite{ortega2018graph},\cite{shuman2013emerging}. 
Based on a graph Fourier space defined by the eigenspace of the representing adjacency
or Laplacian matrix, 
GSP filters have found applications in 
practice, including bridge health monitoring \cite{chen2014semi}, point cloud denoising \cite{schoenenberger2015graph}, and image classification \cite{sandryhaila2013discrete1}. 
Leveraging graph Fourier transform \cite{sandryhaila2013discrete}, graph wavelet \cite{hammond2011wavelets} and graph spectral convolution \cite{shuman2012windowed} 
can extract additional features from graph signals. For example,
graph convolutional filters have been successful in 
edge detection 
and video segmentation \cite{zhang2020hypergraph}. 

Although GSP-based spectral filtering has demonstrated successes in 
a variety of applications, it still suffers from the high-complexity 
of spectrum computation and the need to 
select suitable propagation models. 
To efficiently extract signal features and integrate traditional GSP 
within the machine learning framework, 
graph convolutional networks (GCN) \cite{kipf2017semi} have
been developed for semi-supervised classification problems. 
Approximating graph spectral convolution with first-order Chebyshev expansions, 
GCN has been effective in such learning tasks. Furthermore, 
different GCN-related graph learning architectures, 
including personalized propagation 
of neural predictions (PPNP)\cite{klicpera2018predict} and N-GCN\cite{abu2018n}, 
have also been developed to process graph-represented datasets.
However, some limitations remain with the traditional GCN based on Chebyshev expansions. For example, 
traditional GCN 
requires strong assumptions on maximum eigenvalues and Chebyshev coefficients 
for approximating spectral convolution, 
at the cost of possible information loss 
when compared against basic
convolutional filters. 
Furthermore,  systematic selection and design of graph representations for GCN remain elusive.

Our goal is to improve GCN by exploring its 
relationship with GSP.
Specifically in this work, we explore the 
process from spectrum wavelet 
to vertex propagation, and investigate 
alternative designs for 
graph convolutional networks. Our contributions can be summarized as follows:
 \begin{itemize}
 	\item We revisit graph spectral convolution in GSP and define
conditions for approximating spectrum wavelet via propagation in the vertex domain.
These conditions could provide insights to design GCN layers.
 	\item  We propose alternative propagation models for the GCN layers and 
 	develop a Taylor-based graph convolutional networks (TGCN) based on the aforementioned approximation conditions.
 \item We illustrate the effectiveness of the proposed frameworks 
 over several well-known datasets in comparison with
 other GCN-type and graph-based methods in node
 signal classification.
 \item We also provide an interpretability discussion on the use of Taylor expansion, together with guidelines on selecting
 suitable graph representations for GCN layers.
 \end{itemize}

In terms of the manuscript organization, we first provide an 
overview on graph-based tools and review the fundamentals 
of GCN in Section II and Section III, respectively. Next, 
we present the theoretical motivation and basic design of the 
Taylor-based graph convolutional networks (TGCN) in Section IV. 
Section V discusses the interpretability 
of Taylor expansion in approximation, and highlights
the difference between TGCN and traditional GCN-based methods. 
We report the experimental results of the proposed TGCN framework
on different datasets in Section VI, before concluding
in Section VII.

\section{Related Work}
In this section, we provide an overview on state-of-the-art graph signal processing (GSP) 
and graph convolutional networks (GCN). 

\subsection{Graph Signal Processing} Graph signal processing (GSP) has emerged as an exciting and promising new tool for 
processing large datasets with complex structures \cite{ortega2018graph,shuman2013emerging}. 
Owing to its power to extract underlying relationships among signals,
GSP has achieved significant success in generalizing
traditional digital signal processing (DSP) and 
processing datasets with complex underlying features. 
Modeling data points and their interactions as graph nodes
and graph edges, respectively,
graph Fourier space could be defined according to the eigenspace 
of a graph representing matrix, 
such as the Laplacian or adjacency matrix, 
to facilitate data processing operations such as
denoising \cite{wagner2006distributed}, 
filter banks \cite{narang2010local}, and compression \cite{zhu2012approximating}. 
The framework of GSP can be further generalized over
the graph Fourier space to include sampling theory \cite{chen2015discrete}, 
graph Fourier transform \cite{sandryhaila2013discrete}, frequency analysis \cite{sandryhaila2014discrete}, graph filters \cite{sandryhaila2013discrete1}, graph wavelet \cite{hammond2011wavelets} and graph stationary process \cite{marques2017stationary}.
In addition, GSP has also been considered for
high-dimensional geometric signal processing, such as hypergraph signal processing \cite{zhang2019introducing} and topological signal processing \cite{barbarossa2020topological}.

\subsection{Graph Convolutional Networks} 
Graph-based learning machines 
have become useful tools in data analysis. 
Leveraging
graph wavelet processing \cite{hammond2011wavelets}, graph convolutional networks (GCN) approximate the spectral wavelet convolution via first-order Chebyshev expansions \cite{kipf2017semi} and have demonstrated notable successes in 
semi-supervised learning tasks. Recent works,
e.g., \cite{klicpera2018predict}, 
have developed customized propagation of 
neural predictions (PPNP) to 
integrate PageRank \cite{page1999pagerank} with GCNs. 
Other typical graph-based learning machines include GatedGCN \cite{bresson2017residual},  GraphSAGE \cite{hamilton2017inductive}, Gaussian Mixture Model Network (MoNet) \cite{monti2017geometric}, Graph Attention Networks (GAT) \cite{velivckovic2017graph}, Differential Pooling (DiffPool) \cite{ying2018hierarchical}, Geom-GCN \cite{pei2020geom}, Mixhop \cite{abu2019mixhop}, Diffusion-Convolutional Neural Networks (DCNN) \cite{atwood2016diffusion}, and Graph Isomorphism Network (GIN) \cite{xu2018powerful}. 
For additional information, interested readers are referred to 
an extensive 
literature review \cite{bronstein2017geometric} and a 
survey paper \cite{dwivedi2020benchmarking}.
\begin{table}[]
	\centering
	\caption{Notations and Definitions}
	\begin{tabular}{l|l}
		\hline
		Notation & Definition \\\hline
		$\mathcal{G}=(\mathcal{V,E})$ &  The undirected graph          \\
		$\mathcal{V}$& The set of nodes in the graph $\mathcal{G}$     \\
	    $\mathcal{E}$& The set of edges in the graph $\mathcal{G}$     \\ 
	    $\mathbf{A}$& The adjacency matrix   \\
	    $\mathbf{L}$& The Lapalcian matrix   \\  
	    $\mathbf{P}$& The general propagation matrix of the graph $\mathcal{G}$  \\
	    $\mathbf{D}$& The diagonal matrix of node degree   \\   
	    $\lambda$ & The eigenvalue of the propagation matrix\\
	     $\lambda_{max}$ & The maximal eigenvalue\\
	    $\mathbf{V}$ & The spectral matrix with eigenvectors of $\mathbf{P}$ as each column\\
	    $\mathbf{x}$ & The graph signal vector\\
	     $\mathbf{X}^{(l)}$ & The feature data at $l$ layer\\
	     $\mathbf{I}$ & The identity matrix\\
	    \hline
	\end{tabular}
\label{nota}
\end{table}

\section{Graph Wavelet and Graph Convolutional Networks}\label{GWT}
In this section, we first review the fundamentals of graph spectral convolution and wavelets, necessary for the development of propagation models of the GCN layers. We will 
then briefly introduce the structures of traditional GCN \cite{kipf2017semi}. For convenience, some of the important notations and definitions are illustrated in Table \ref{nota}.

\subsection{Graph Spectral Convolution and Wavelet-Kernels}

An undirected graph $\mathcal{G}=(\mathcal{V,E})$ with $N=|\mathcal{V}|$ nodes
can be represented by a representing matrix (adjacency/Laplacian) decomposed as $\mathbf{P}=\mathbf{V\Sigma V}^T\in\mathbb{R}^{N\times N}$, where the eigenvectors $\mathbf{V}=\{\mathbf{f}_1,\mathbf{f}_2,\cdots,\mathbf{f}_N\}$ form the 
graph Fourier basis and the eigenvalues $\lambda_i$'s represent graph frequency \cite{sandryhaila2013discrete}. 

In GSP \cite{shuman2012windowed}\cite{shi2019graph}, 
graph Fourier transform of convolution between two signals is a
product between their respective Fourier transforms denoted by $\diamond$, i.e., 
\begin{equation}\label{conv}
\mathbf{x}\diamond \mathbf{y}=\mathcal{F}_C^{-1}(\mathcal{F}_C(\mathbf{x})\circ\mathcal{F}_C(\mathbf{y})),	
\end{equation}
where $\mathcal{F}_C(\mathbf{x})=\mathbf{V}^T\mathbf{x}$ refers to the graph Fourier transform (GFT) of signals $\mathbf{x}$, $\mathcal{F}_C^{-1}(\hat{\mathbf{x}})=\mathbf{V}\hat{\mathbf{x}}$ is the inverse GFT and $\circ$ is the Hadamard product. This definition generalizes the 
property 
that convolution in the vertex domain is equivalent to product in
the corresponding graph spectral domain.

In \cite{hammond2011wavelets}, the graph wavelet transform is defined according to
graph spectral convolution. Given a spectral graph wavelet-kernel 
$\hat{\mathbf{g}}=[g(\lambda_1), g(\lambda_2), \cdots, g(\lambda_N)]^T$ with kernel function $g(\cdot)$, 
the graph wavelet operator is defined as 
\begin{align}
T_g(\mathbf{x})&=\mathbf{V}(\hat{\mathbf{g}}\circ (\mathbf{V}^T\mathbf{x}))\\
&=\mathbf{V}
\begin{bmatrix}
&g(\lambda_1) &\cdots &0\\
&0 &\ddots &0\\
&0 &\cdots &g(\lambda_N)
\end{bmatrix}
\mathbf{V}^T\mathbf{x}.\label{wavelet}
\end{align}
Note that graph wavelet can be interpreted as a graph convolutional filter with a spectrum wavelet-kernel $\hat{\mathbf{g}}$. Depending on the datasets and applications, 
different kernel functions may be utilized 
in (\ref{wavelet}). 

\subsection{Graph Convolutional Networks and Their Limitations}

To overcome the complexity 
for computing the spectrum matrix $\mathbf{V}$ and the 
difficulty of seeking suitable wavelet-kernel functions, 
one framework of GCN developed in \cite{kipf2017semi}
considers a first-order Chebyshev expansion. 
Considering Chebyshev polynomials $T_K(x)$ up to $K^{th}$ orders and the Laplacian matrix as the propagation matrix, the graph
convolutional filter with wavelet-kernel $\hat{\mathbf{g}}$ is approximated by 
\begin{equation}\label{cheby}
	T_g(\mathbf{x})\approx\sum_k\theta_kT_k(\tilde{\mathbf{L}})\mathbf{x},
\end{equation}
where $\tilde{\mathbf{L}}=2 \mathbf{L}/\lambda_{\max}-\mathbf{I}_N$. With careful choice of $\lambda_{\max}$ and parameters $\theta_k$, the graph convolutional filter 
can be further approximated by the 1st-order Chebyshev expansion
\begin{equation}
	T_g(\mathbf{x})\approx \theta(\mathbf{I}_N+\mathbf{D}^{-\frac{1}{2}}\mathbf{A}\mathbf{D}^{-\frac{1}{2}})\mathbf{x},
\end{equation}
where $\mathbf{D}$ is the diagnal matrix of node degree.
From here, by generalizing the approximated graph convolutional filter to a signal $\mathbf{X}\in\mathbb{R}^{N\times C}$ with $C$ features for each node, the filtered signals 
can be written as 
\begin{equation}\label{DAD}
\mathbf{Z}=\tilde{\mathbf{D}}^{-\frac{1}{2}}\tilde{\mathbf{A}}\tilde{\mathbf{D}}^{-\frac{1}{2}}\mathbf{X}\mathbf{\Theta},
\end{equation} 
where $\tilde{\mathbf{A}}=\mathbf{A}+\mathbf{I}_N$, $\tilde{D}_{ii}=\sum_j\tilde{A}_{ij}$, and $\mathbf{\Theta}\in\mathbb{R}^{C\times F}$ is the parameter matrix. Furthermore, by integrating the nonlinear functions within 
the approximated convolutional filters, a two-layer GCN can be designed with message propagation as
\begin{align}\label{gcn_prop}
\mathbf{Z}_{GCN}=&\mathrm{softmax}
\left(\tilde{\mathbf{D}}^{-\frac{1}{2}}\tilde{\mathbf{A}}\tilde{\mathbf{D}}^{-\frac{1}{2}}
\right.\nonumber\\
&\left.\mathrm{RELU}(\tilde{\mathbf{D}}^{-\frac{1}{2}}\tilde{\mathbf{A}}\tilde{\mathbf{D}}^{-\frac{1}{2}}\mathbf{XW}^{(0)}))\mathbf{W}^{(1)}\right),
\end{align}
where $\mathbf{W}^{(0)}\in\mathbb{R}^{N\times H}$ and $\mathbf{W}^{(1)}\in\mathbb{R}^{H\times C}$ are the parameters for the $H$ hidden units. Here we use standard
terminologies of ``softmax'' and ``RELU'' from
deep learning neural networks. 

Although GCN has achieved success in some
applications, some drawbacks remain. First, it relies on 
several strong assumptions
to approximate the original convolutional filters. 
For example, $\lambda_{\max}=2$ are used 
to approximate in implementation due to the range of variables in Chebyshev expansions, and the Chebyshev coefficients are set to $\theta_1=-\theta_0=-\theta$ to obtain 
Eq.~(\ref{DAD}). These assumptions may compromise
the efficacy of spectral convolution. Second, the graph representation $\tilde{\mathbf{D}}^{\frac{1}{2}}\tilde{\mathbf{A}}\tilde{\mathbf{D}}^{\frac{1}{2}}$ may not 
always be the optimal choice while the Chebyshev 
approximation limits the type of representing matrix. 
We provide a more detailed discussion in Section V. 
In addition, it remains unclear as to how best to derive a suitable kernel-function $\hat{\mathbf{g}}$ and its approximation.
Moreover, insights in terms of interpretability is
highly desirable from spectral wavelet convolution 
to vertex propagation.

To explore alternatives in designing propagation model for 
GCNs, we focus 
on the process between graph spectral wavelet-kernels 
and propagation in the vertex domain. 
We will further propose alternative propagation
models for GCNs.

\section{Taylor-based Graph Convolutional Networks}
In this section, we investigate conditions needed
for approximating the spectral convolution via vertex 
propagation.
Next, we propose alternative propagation models for graph
convolution layers based on Taylor expansion, where the 
general convolutional filter can be written as 
\begin{equation}
\mathbf{Z}=G_\alpha(\mathbf{P})\mathbf{X}\mathbf{\Theta},
\end{equation}
where $G_\alpha(\mathbf{P})$ is a polynomial function with 
parameter $\alpha$, $\mathbf{P}$ is the representing matrix 
of the graph, and $\mathbf{\Theta}$ are parameters of 
feature projection.

\subsection{Approximation of Spectral Convolution}
We first present the theoretical motivation for designing
a polynomial-based propagation model, and its 
relationship to the graph spectral wavelets.
For a polynomial filter in GSP, let $\mathbf{P}$ be the representing (adjacency/Laplacian) matrix.
We can obtain the following property.
\begin{lemma}
	Given a GSP polynomial filter $\mathbf{H}=h(\mathbf{P})=\sum_k \alpha_k \mathbf{P}^k$, 
the filtered signals are calculated by
	\begin{equation}
	\mathbf{Hx}=h(\mathbf{P})\mathbf{x}=\sum_{r=1}^{N}h(\lambda_r) \mathbf{f}_r(\mathbf{f}_r^{\mathrm{T}}\mathbf{x}),
	\end{equation}
	where $\mathbf{f}_r$'s are the graph spectrum and $\lambda_r$'s are the eigenvalues of $\mathbf{P}$ related to graph frequency.
\end{lemma}
\begin{proof}
	Let $\mathbf{V}=[\mathbf{f}_1,\cdots,\mathbf{f}_N]$ and $\Sigma={\rm diag}([\lambda_1,\cdots,\lambda_N])$. Since $\mathbf{V}^\mathrm{T}\mathbf{V}=\mathbf{I}$, 
	we have
\begin{align}
	\mathbf{P}^k\mathbf{x}&=\underbrace{\mathbf{V}\Sigma\mathbf{V}^\mathrm{T}\mathbf{V}\Sigma\mathbf{V}^\mathrm{T}\cdots \mathbf{V}\Sigma\mathbf{V}^\mathrm{T}}_{k\quad {\rm times}}\mathbf{x}\nonumber\\
	&=\mathbf{V}\Sigma^k\mathbf{V}^\mathrm{T}\mathbf{x}\\
	&=\sum_{r=1}^N \lambda_r^k (\mathbf{f}_r^{\mathrm{T}}\mathbf{x})\mathbf{f}_r.
\end{align}
	Since $\mathbf{H}=h(\mathbf{P})=\sum_k \alpha_k\mathbf{P}^k$ is a polynomial graph filter, we can directly obtain
	\begin{align}
	\mathbf{Hx}&=\sum_k \sum_{r=1}^N \alpha_k\lambda_r^k  (\mathbf{f}_r^{\mathrm{T}}\mathbf{x})\mathbf{f}_r\\
	&=\sum_{r=1}^N h(\lambda_r)  (\mathbf{f}_r^{\mathrm{T}}\mathbf{x})\mathbf{f}_r.
	\end{align}
	
\end{proof}
This lemma shows that
the response of the filter to an exponential is the same exponential amplified by a gain that is the frequency response of the filter at the frequency of the exponential \cite{ortega2018graph}. The exponentials are the eigenfunctions/eigenvectors, 
similar to complex exponential signals in linear systems.

Looking into the graph wavelet convolutional filter in Eq. (\ref{wavelet}), the wavelet-kernel function $g(\cdot)$ operates to modify frequency coefficients $\lambda_r$'s. 
Thus, we have the following property of transferring spectrum wavelet to vertex propagation.

\begin{theorem}
Given a polynomial wavelet kernel function $g(\cdot)$, the GSP convolutional filter on signal $\mathbf{x}$ is calculated as
	\begin{equation}
	T_g(\mathbf{x})=g(\mathbf{P})\mathbf{x}.
	\end{equation}
\end{theorem}
\begin{proof}
	Since the convolution filter $T_g(\mathbf{x})$ can be written in
	\begin{equation}
		T_g(\mathbf{x})=\sum_{r=1}^{N}g(\lambda_r) \mathbf{f}_r(\mathbf{f}_r^{\mathrm{T}}\mathbf{x}),
	\end{equation}
the proof is straightforward by invoking with \textit{Lemma} 1.
\end{proof}
This theorem indicates that we can bypass computing
the spectrum by implementing the convolution directly 
in vertex domain, since 
the wavelet kernel $g(\cdot)$ is polynomial or 
can be approximated by a polynomial expansion. We can see that the Chebyshev expansion is a special case 
of \textit{Theorem 1}. In addition to
Chebyshev expansion, Legendre \cite{pons1998legendre} and Taylor \cite{linnainmaa1976taylor} expansions can also approximate the spectral convolution. In addition, other polynomial design on the wavelet-function $g(\cdot)$ are also possible.

\subsection{Taylor-based Propagation Model}\label{model}
We now provide alternative propagation models for the GCN layers based on Taylor expansions,
with which the wavelet-kernel function $g(x)$ can be approximated via
\begin{equation}\label{taylor}
g(x)\approx \sum_{k=0}^K  \theta_k(x-a)^k.
\end{equation}
Here, $\theta_k={g^{(k)}(a)}/{n!}$ is the expansion coefficients.

Since the Taylor approximation of $g(x)$ in Eq. (\ref{taylor}) is a polynomial function of the variable $x$ which meets the condition in \textit{Theorem} 1, the graph spectral convolutional filter can be approximated as
\begin{align}\label{TGCN}
T_g(\mathbf{x})
\approx\sum_{k=0}^K \theta_k(\mathbf{P}-{\rm diag}(\bm{\Phi}))^k\mathbf{x},
\end{align}
where $\bm{\Phi}$ is a generalization of the parameter $a$. We will discuss further the intuition of applying Taylor expansion and its difference with Chebyshev approximation in Section V.

We can develop different models based on Eq. (\ref{TGCN}) to 
develop the Taylor-based GCN (TGCN).
The $\mathbf{P}$ matrix here can be any practical 
graph representing matrix
used to capture overall information of the graph. 
For example, typical representing matrices include the 
adjacency matrix $\mathbf{A}$, the Laplacian matrix $\mathbf{L}$,
or the normalized propagation matrix $\tilde{\mathbf{D}}^{-\frac{1}{2}}\tilde{\mathbf{A}}\tilde{\mathbf{D}}^{-\frac{1}{2}}$. 
Section \ref{exper} provides
further discussions on the selection of 
representing matrix.

We now provide several types of TGCN design.

\textbf{Type-1 First-Order TGCN}: Similar to the traditional GCN, we first consider 
TGCN based on the first-order Taylor expansions with a simpler diagonal matrix ${\rm diag}(\bm{\Phi})=\phi\mathbf{I}_N$. Letting $K=1$,
Eq. (\ref{TGCN}) can be written as
\begin{align}
T_g(\mathbf{x})&\approx[(\theta_0-\theta_1\phi)\mathbf{I}_N+\theta_1\mathbf{P}]\mathbf{x}\\
&=\theta'(\mathbf{P}+\alpha \mathbf{I}_N)\mathbf{x},
\end{align}
where $\theta'=\theta_1$ and $\alpha=\frac{\theta_0-\theta_1\phi}{\theta_1}$ are the new parameters for the convolutional filter. As a result, the GCN layer with generalized signal $\mathbf{X}\in\mathbb{R}^{N\times C}$ can be designed as 
\begin{equation}
\mathbf{X}^{(l+1)}=(\mathbf{P}+\alpha_l \mathbf{I}_N)\mathbf{X}^{(l)}\bm{\Theta}_l
\end{equation}
where $\alpha_l$ and $\bm{\Theta}_l$ are the trainable variables for the $l^{th}$ layer. 

\textbf{Type-2 First-Order TGCN}: We also consider more 
general diagonal matrix in place of $\alpha\mathbf{I}_N$, i.e.,
\begin{equation}
\mathbf{X}^{(l+1)}=(\mathbf{P}+{\rm diag}(\bm{\beta}_l))\mathbf{X}^{(l)}\bm{\Theta}_l,
\end{equation}
where $\bm{\beta}_l$ and $\bm{\Theta}_l$ are the parameters of the $l^{th}$ layer. Here, the self-influence for each node varies from node to node, 
whereas each node affects itself equivalently in the type-1 first-order TGCN model.

\textbf{Type-3 $k^{th}$-Order TGCN}: We also consider the higher-order polynomial propagation models for each layer. To avoid overfitting and reduce the complexity, we 
require $\theta_k=\theta$ for all $k$
and the diagonal matrix as $\alpha \mathbf{I}_N$ 
in Eq. (\ref{TGCN}). Then the resulting TGCN layer becomes
\begin{equation}
\mathbf{X}^{(l+1)}=[\sum_k(\mathbf{P}+\alpha_{l}\mathbf{I}_N)^k]\mathbf{X}^{(l)}\bm{\Theta}_l.
\end{equation}
Compared to the 1st-order approximation,  
the higher-order approximation contains
more trainable parameters and computations, 
resulting in higher implementation complexity.

\textbf{Type-4 $k^{th}$-Order TGCN}: More general TGCN layers can be designed without 
requiring  $\theta_k=\theta$ for all $k$ as follows:
\begin{equation}
\mathbf{X}^{(l+1)}=\sum_k[(\mathbf{P}+\alpha_{l}\mathbf{I}_N)^k\mathbf{X}^{(l)}\bm{\Theta}_{l,k}].
\end{equation}
Here, we only consider simple diagonal with one parameter $\alpha$ in 
the higher-order polynomials to avoid overfitting and high complexity. 
We will provide some insights into
the choice of different approximation models in 
Section \ref{exper}.

\section{Discussion}
In this section, we discuss the interpretability of
the Taylor approximation of wavelet-kernels, and illustrate its differences with the existing GCNs.

\subsection{Interpretation of the Graph Convolution Approximation}

To understand the connection between the graph convolution filters and the approximated GCNs, we start from the basic graph
wavelet operator in Eq. (\ref{wavelet}). Given the definition of graph convolution in Eq. (\ref{conv}) and the graph spectrum matrix $\mathbf{V}=[\mathbf{f}_1,\mathbf{f}_2,\cdots,\mathbf{f}_N]$, the graph wavelet operator is a convolution-based filter on the signal $\mathbf{x}$ with a parameter vector $\bm{\sigma}$, and is 
denoted by
\begin{equation}
	T_g(x)=\bm{\sigma}\diamond\mathbf{x},
\end{equation}
where the graph Fourier transform of $\bm{\sigma}$ requires $\mathbf{V}^T\bm{\sigma}=\hat{\mathbf{g}}$. Suppose that $\mathbf{f}_i$ is the corresponding eigenvector of $\lambda_i$. Each term in the wavelet-kernel (i.e., $g(\lambda_i)=\mathbf{f}_i^T\bm{\sigma}$) embeds the information of graph spectrum with the parameters $\bm{\sigma}$ by the function $g(\cdot)$. Thus, the optimization on the wavelet-kernel $g(\cdot)$ is equivalent to finding suitable parameters for the 
convolution filter to achieve the goals, such as
minimizing the error between the filtered signal 
and its labels.

Computing the exact graph spectra can be time-consuming.
Nevertheless, 
GCNs and TGCNs approximate the graph-kernel by 
polynomial expansions, i.e., 
\begin{equation}
	T_g(\mathbf{x})
	\approx\sum_k \theta_kT_k(\mathbf{P}){x},
\end{equation}
where $\theta_k$ is the expansion coefficients and $T_k(\cdot)$ is the $k$-th term of polynomials. Here, the $\theta_k$ is a function of $g(\cdot)$, i. e., 
\begin{equation}
	\theta_k={g^{(k)}(a)}/{n!}
\end{equation}
for Taylor polynomials, and
\begin{equation}
	 \theta_k=\frac{2}{\pi}\int_{-1}^{-1} \cos(k\theta)g(\cos(\theta))d\theta
\end{equation}
for Chebyshev polynomials \cite{hammond2011wavelets}.

Since $T_k(\mathbf{P})$ is already determined for each type of expansions, the optimization on the wavelet-kernel function $g(\cdot)$ is transformed into estimating the polynomial coefficients $\theta_k$, i.e., $\bm{\Theta}$, in the GCN propagation layers. More specifically, although we may need prior knowledge on 
the derivatives of $g(\cdot)$ for Taylor coefficients, ${g^{(k)}(a)}$ can be reparametrized as the parameters $\theta_k$ 
of the convolution filter, given our goal to optimize the 
function $g(\cdot)$ by using deep learning networks.
The information of the Taylor expansions remains in the polynomials terms, i.e.,
\begin{equation}
T_k(\mathbf{P})=(\mathbf{P}-{\rm diag}(\bm{\Phi}))^k,
\end{equation}
which differs from the Chebyshev expansions.

\subsection{Comparison with Chebyshev-based GCNs}
Now, we compare the differences between Chebyshev-based GCNs and Taylor-based GCNs to illustrate the benefits of Taylor expansions.

For the Chebyshev expansion, variable $x$ is bounded within
$[-1, 1]$. In the graph wavelet-kernel, each $\lambda$ is 
within [0, $\lambda_{\max}]$ for the Laplacian matrix. 
A simple transformation $\lambda=a(y+1)$, with $a=\lambda_{max}/2$ \cite{hammond2011wavelets} can change variables 
from $\lambda$ to $y$:
\begin{equation}
	y=\frac{2\lambda}{\lambda_{max}}-1,
\end{equation}
which accounts for the use of graph representation, i. e., 
\begin{equation}
	\tilde{\mathbf{L}}=2 \mathbf{L}/\lambda_{\max}-\mathbf{I}_N,
\end{equation}
in Eq. (\ref{cheby}) for Chebyshev-based GCNs.
Unlike Chebyshev expansion, Taylor polynomials do not limit
variable $x$ to an interval (without 
using $\lambda_{\max}$). For this reason, TGCN admits
a more flexible design of graph representation 
$\mathbf{P}$ without limiting the intervals of eigenvalues. 
There is no need to set the value of the largest eigenvalue 
or normalize the graph representation when implementing
TGCNs. We shall test different graph representations when 
we present the experiment results next. 

In addition, the Taylor polynomial gives a more unified
simple design for polynomial terms ${T}_k(\mathbf{P})$. 
In Chebyshev polynomials, each polynomial term is recursive
from its previous terms, thereby making it less expedient
to implement higher-order GCN. However, Taylor polynomials take the same form regardless of $a$
and are regular. The additional parameters $\bm{\Phi}$ also provide benefits for TGCN. In signal propagation,
parameters $\bm{\Phi}$ can be interpreted as the self-influence
from the looping effects. In practical applications, 
such self-influence does occur and may be less obvious 
within the data. For example, 
in the citation networks, the work from highly-cited 
authors may have 
greater impact and trigger the appearance of a 
series of related new works on its own, 
which indicates larger self-influence as well as 
higher impact on other works. 
We will illustrate this impact further in Section VI.

\subsection{Single-layer High-order vs. Multi-layer First-order}
From the design of the first-order TGCN, 
multiple layers of the first-order propagation also 
forms a higher-order polynomial design. However, such design 
with multiple layers of the first-order polynomials is 
different from the single-layer high-order propagation. 
Suppose that a single-layer $k^{th}$-order polynomial 
convolutional filter is written as
\begin{equation}
\mathbf{Z}_k=\sum_k\alpha_k(\mathbf{P}+{\rm diag} (\bm{\beta}))^k\mathbf{X}\bm{\Theta}
\end{equation}
and the first-order polynomial is
\begin{equation}
\mathbf{Z}_1=\alpha(\mathbf{P}+{\rm diag}(\bm{\beta}))\mathbf{X}\bm{\Theta}. 
\end{equation}

For a $k$-layer first-order polynomial convolutional filter, the filtered result can be written as
\begin{align}
\mathbf{Z}^{(k)}&=\alpha_1\cdots\alpha_k(\mathbf{P}+{\rm diag}(\bm{\beta}))^k\mathbf{X}\bm{\Theta}_1\cdots\bm{\Theta}_k\\
&=\alpha'(\mathbf{P}+{\rm diag}(\bm{\beta}))^k\mathbf{X}\bm{\Theta}',
\end{align}
which is one term in the single-layer $k^{th}$-order polynomial convolutional filter. Thus, the multi-layer first-order TGCN is a special case of single-layer higher-order polynomials. Since the high-order polynomial designs have already embedded the high-dimensional propagation of signals over the graphs, we usually apply the single-layer design instead of multi-layer for higher-order TGCNs in implementation to reduce complexity.

\section{Experiments}\label{exper}
We now test the TGCN models in different classification 
experiments. We first measure the influence of depth and 
polynomial orders in different types of TGCNs in citation networks. 
Next, we experiment with different representing matrices and 
propagation models to explore the choice of suitable layer 
design and graph representations. We also report comparative
results with other GCN-like methods in node classification
and demonstrate the practical competitiveness of our newly
proposed framework.

\subsection{Evaluation of Different TGCN Designs}\label{reu}
In this subsection, we evaluate different designs of TGCN to show its overall performance in node classification of citation networks. 
Additional comparisons with other existing methods 
are provided in Section \ref{compare}.
\subsubsection{Experiment Setup}
We first provide the TGCN experiment setup.

\textbf{Datasets}: We use three citation network datasets for validation, i.e., Cora-ML \cite{mccallum2000automating,bojchevski2017deep}, Citeseer \cite{sen2008collective}, and Pubmed \cite{namata2012query}. In these citation networks, published
articles are denoted as nodes and their citation relationships 
are represented by edges. The data statistics 
of these citation networks are summary in Table \ref{sta}. We randomly select the a subset from the original datasets with 10\% training data, 60\% test data and 30\% validation data. 

\begin{table*}[t]
	\caption{Data Statistics}
	\centering
	\begin{tabular}{|c|c|c|c|c|c|}
		\hline
		Datasets     & Number of Nodes     & Number of Edges &  Number of Features &Number of Classes & Label Ratio \\\hline
		Cora & 2708  & 5728 & 1433 & 7& 0.052    \\\hline
		Citeseer     & 3327 & 4614 & 3703& 6&0.036     \\\hline
		Pumbed    & 19717   & 44325& 500& 3&0.0031  \\\hline
	\end{tabular}
	\label{sta}
\end{table*}

\begin{table*}[t]
	\centering
	\begin{threeparttable}[b]
		\centering
		\caption{Overall Accuracy for Different First-Order Methods (Percent)}
		\begin{tabular}{|c|c|c|c|c|}
			\hline
			Methods&Representing Matrix& Cora & Citeseer & Pubmed\\\hline
			GCN&$\tilde{\mathbf{D}}^{-\frac{1}{2}}\tilde{\mathbf{A}}\tilde{\mathbf{D}}^{-\frac{1}{2}}$&78.7$\pm$2.7&68.7$\pm$2.7&78.8$\pm$3.1\\\hline
			Type-1 First-Order TGCN$^1$&$\tilde{\mathbf{D}}^{-\frac{1}{2}}\tilde{\mathbf{A}}\tilde{\mathbf{D}}^{-\frac{1}{2}}$&79.9$\pm$1.7&70.1$\pm1.5$&79.3$\pm$2.0\\\hline
			Type-1 First-Order TGCN$^2$&$\tilde{\mathbf{D}}^{-\frac{1}{2}}\tilde{\mathbf{A}}\tilde{\mathbf{D}}^{-\frac{1}{2}}$&78.8$\pm$3.4&68.9$\pm$1.4&78.9$\pm$2.5\\\hline
			Type-2 First-Order TGCN&$\tilde{\mathbf{D}}^{-\frac{1}{2}}\tilde{\mathbf{A}}\tilde{\mathbf{D}}^{-\frac{1}{2}}$&$\bm{81.3\pm2.7}$&$\bm{70.3\pm2.6}$&$\bm{79.8\pm2.6}$\\\hline
		\end{tabular}
		\label{ovacc}
		\begin{tablenotes} 
			\item [*] For type-1 first-order TGCN with propagation model $\tilde{\mathbf{D}}^{-\frac{1}{2}}\tilde{\mathbf{A}}\tilde{\mathbf{D}}^{-\frac{1}{2}}$, we adjust parameters both manually and automatically. The results are reported in $^1$ for manual and $^2$ for automatic adjustments, respectively.
		\end{tablenotes}
	\end{threeparttable}
\end{table*}

\begin{figure*}[t]
	\centering
	\subfigure{
		\label{sur1}		
		\centering
		\includegraphics[height=1.5in]{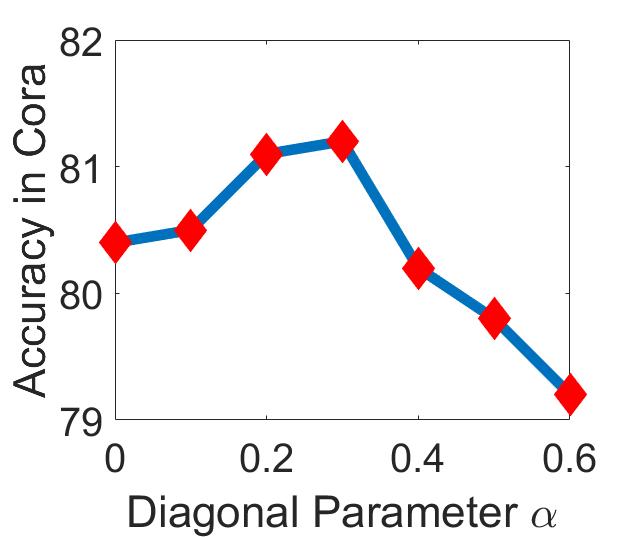}
	}
	\hfill
	\subfigure{
		\label{cyl1}		
		\centering
		\includegraphics[height=1.5in]{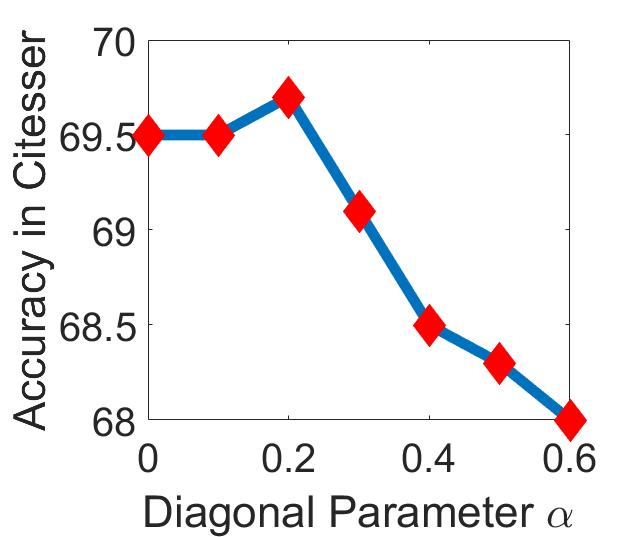}
	}
	\hfill
	\subfigure{
		\label{cub1}
		\centering
		\includegraphics[height=1.5in]{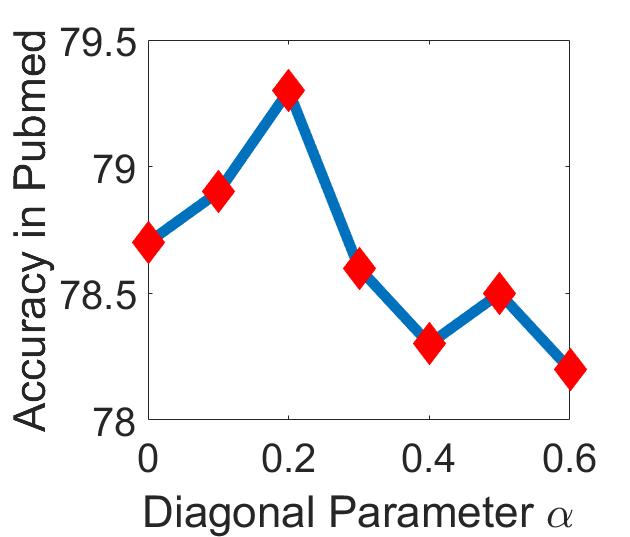}
	}
	\centering
	\caption{Optimal $\alpha$ for Type-1 First-Order TGCN.}
	\label{ori}
\end{figure*}

\textbf{Convolution Layer}: For the first-order TGCN, we consider a two-layer structure designed as follows.
\begin{equation}
\mathbf{Z}=\mathrm{softmax}(G_{\Phi_1}(\mathbf{P})\mathrm{RELU}(G_{\Phi_0}(\mathbf{P})\mathbf{X}\mathbf{W}^{(0)})\mathbf{W}^{(1)}),
\end{equation}
where $G_{\Phi_1}(\mathbf{P})$ is the specific type of TGCN propgation model, and $\mathbf{W}^{(0)}\in\mathbb{R}^{N\times H}$ together with $\mathbf{W}^{(1)}\in\mathbb{R}^{H\times C}$ are the parameters of the $H$ hidden units. 
For the higher-order TGCN, we consider a single-layer structure, i.e.,
\begin{equation}
\mathbf{Z}=\mathrm{softmax}(G_{\Phi}(\mathbf{P})\mathbf{XW}).
\end{equation}
When training the parameters, we let the
neural networks learn the diagonal parameters $\bm{\beta}$ for 
type-2 first-order TGCN, and $\alpha$ for higher-order TGCN. 
We applied Adam optimizer \cite{kingma2014adam} for network training.
For type-1 first-order TGCN, we apply both manual and automatic adjustments
on the diagonal parameters $\alpha$. We train the projection parameter $\mathbf{W}$ 
for a variety of TGCNs. For graph representing matrices, we first apply the normalized $\tilde{\mathbf{D}}^{-\frac{1}{2}}\tilde{\mathbf{A}}\tilde{\mathbf{D}}^{-\frac{1}{2}}$ to measure the effects of layer depth, polynomial orders, and propagation models, respectively. We then test the
results of  different representing matrices $\mathbf{P}$ to 
gain insights and guidelines on how to select suitable 
graph representations.

\textbf{Implementation}: Let $\mathcal{V}_l$ be the set of labeled examples and $Y_i$ denote
the labels. We evaluate the cross-entropy error over all labeled 
examples to train parameters, i.e.,
\begin{equation}
\mathcal{L}=-\sum_{i\in\mathcal{V}_l}\sum_{j=1}^{L}Y_{ij}\ln Z_{ij}.
\end{equation}

\textbf{Hyperparameter}: For fair comparison of different designs of TGCN, we use the similar hyperparameters, with 
dropout rate $d=0.5$, learning rate $r=0.01$ and weight decay $w=5\times 10^{-4}$. 
For two-layer TGCNs, we let the number of hidden units be
$H=40$. For the higher-order TGCNs, we use fewer 
hidden units to reduce the complexity.
\begin{table*}[t]
	\centering
	\begin{threeparttable}[b]
		\centering
		\caption{Accuracy for Higher-Order Propagation Model (Percent)}
		\begin{tabular}{|c|c|c|c|c|}
			\hline
			Num of Layers&Polynomial Order $k$& TGCN Type & Cora & Citeseer\\\hline
			2-Layer&$1^{st}$-order&Type-1&81.4&70.1\\\hline
			2-Layer&$1^{st}$-order&Type-2&\textbf{81.5}&\textbf{70.5}\\\hline
			2-Layer&$2^{nd}$-order&Type-3&79.5&65.7\\\hline
			2-Layer&$2^{nd}$-order&Type-4&79.3&66.9\\\hline
			1-Layer&$1^{st}$-order&Type-1&75.6&67.3\\\hline
			1-Layer&$2^{nd}$-order&Type-3&78.4&69.0\\\hline
			1-Layer&$3^{rd}$-order&Type-3&79.1&68.4\\\hline
			1-Layer&$2^{nd}$-order&Type-4&76.3&67.9\\\hline
			1-Layer&$3^{rd}$-order&Type-4&78.6&70.2\\\hline
		\end{tabular}
		\label{ovacc1}
		\begin{tablenotes} 
			\item [*] For each method, we test with the representing matrix $\tilde{\mathbf{D}}^{-\frac{1}{2}}\tilde{\mathbf{A}}\tilde{\mathbf{D}}^{-\frac{1}{2}}$.
		\end{tablenotes}
	\end{threeparttable}	
\end{table*}

\begin{figure*}[t]
	\centering
	\subfigure{
		\centering
		\includegraphics[height=1.4in]{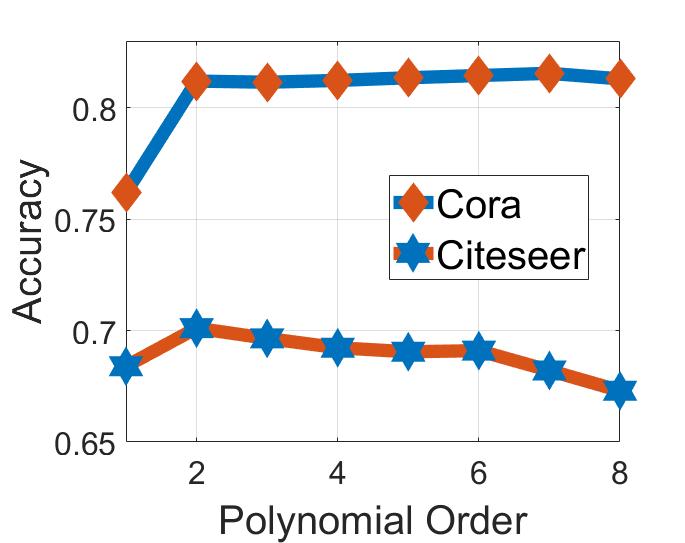}
	}
	\hspace{-5mm}
	\subfigure{
		\centering
		\includegraphics[height=1.4in]{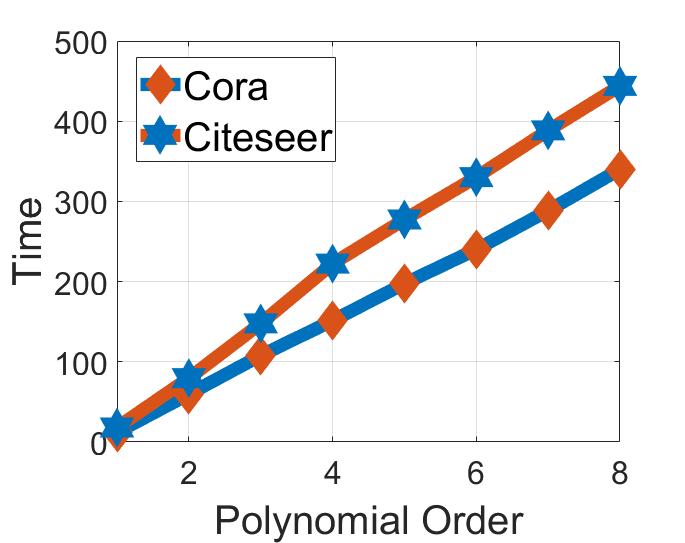}
	}
	\hspace{-5mm}
	\subfigure{
		\centering
		\includegraphics[height=1.4in]{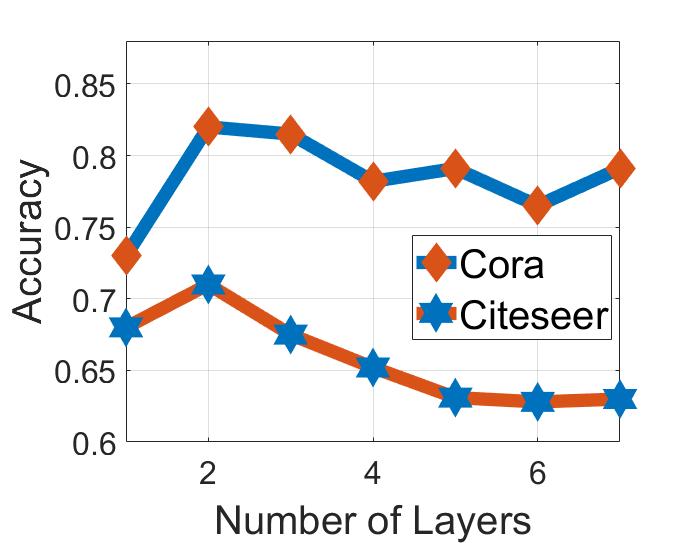}
	}
	\hspace{-5mm}
	\subfigure{
		\centering
		\includegraphics[height=1.4in]{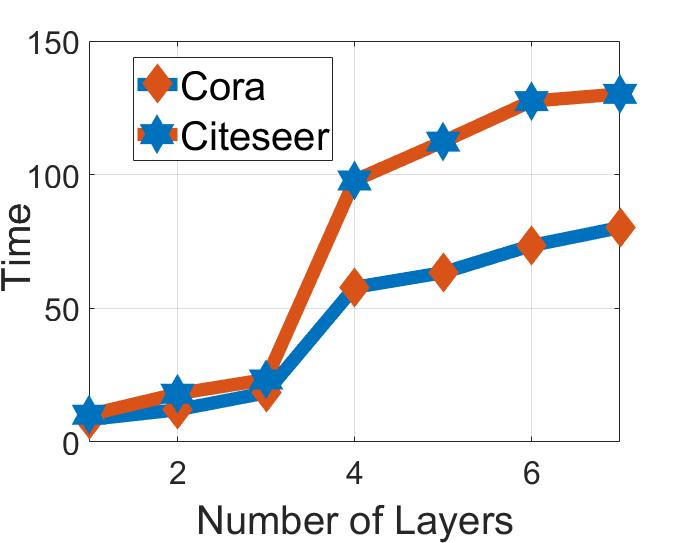}
	}
	\centering
	\caption{Results of Different Polynomial Orders and Network Depth.}
	\label{ori1}
	
\end{figure*}

\begin{figure*}[t]
	\centering
	\subfigure[2-Layer Type-1 TGCN.]{
		\centering
		\includegraphics[height=1.4in]{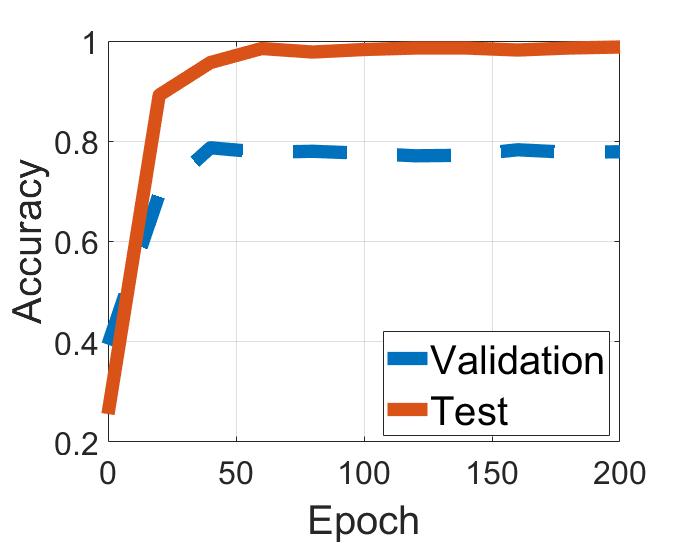}
	}
	\hspace{-5mm}
	\subfigure[2-Layer Type-2 TGCN.]{
		\centering
		\includegraphics[height=1.4in]{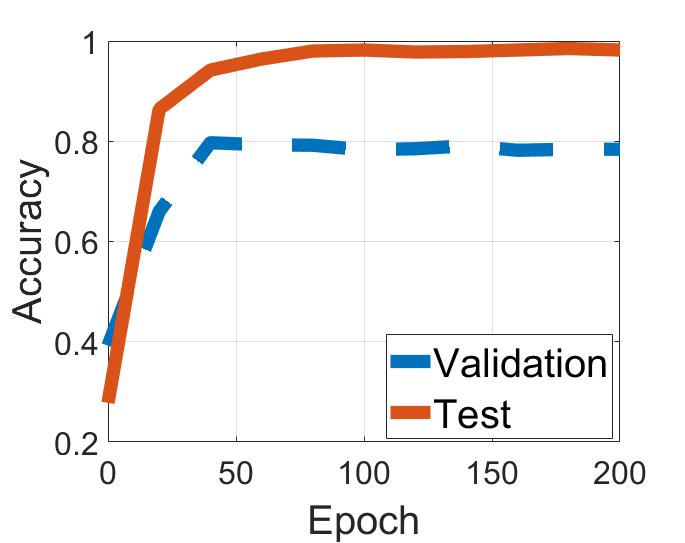}
	}
	\hspace{-5mm}
	\subfigure[1-Layer 2-order Type-3 TGCN.]{
		\centering
		\includegraphics[height=1.4in]{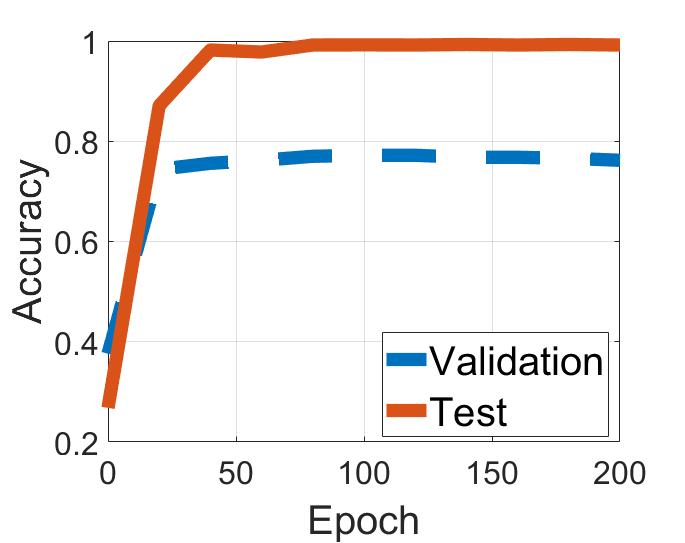}
	}
	\hspace{-5mm}
	\subfigure[1-Layer 2-order Type-4 TGCN.]{
		\centering
		\includegraphics[height=1.4in]{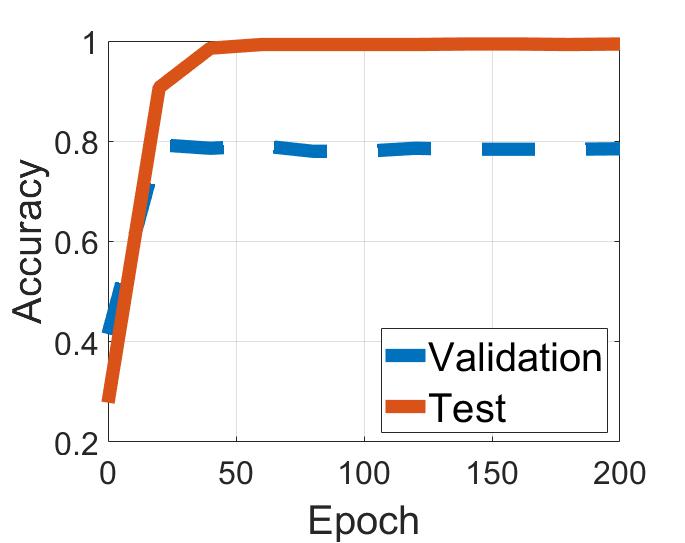}
	}
	\centering
	\caption{Convergence of different TGCN models.}
	\label{ori2}
\end{figure*}

\subsubsection{Performances of Different TGCN Propagation Models}
We first measure the performances of first-order TGCNs 
by comparing different propagation models against the traditional GCNs. 
Note that, the type-1 first-order TGCN degenerates into traditional GCN 
if the diagonal parameter $\alpha=0$. To explore the difference
between GCN 
and type-1 first-order TGCN, we adjust the parameter $\alpha$ both manually and 
automatically.
The overall accuracy is reported in Table \ref{ovacc}. 

For type-1 first-order TGCN, our manual adjustment results in higher accuracy 
than automatic adjustment. This indicates
that the deep learning network may be more susceptible to 
local convergence when learning $\alpha$ by itself. Usually, the optimal $\alpha$ for type-1 TGCN 
would be in $[0.15,0.35]$ as shown in Fig. \ref{ori}, while the TGCN
degenerates
to the traditional GCN for $\alpha=0$. 
Generally, type-2 first-order TGCN has 
a clear advantage in accuracy for all datasets, 
whereas type-1 frist-order TGCN exhibits only marginal improvement 
given suitable choice of diagonal parameters.

We then compare different propagation models with different orders of polynomials under the same experiment setup. We start with 100 
Monte Carlo random initializations and report the average accuracy of each 
model in Table \ref{ovacc1}. The first-order TGCNs generally achieve
superior overall accuracy than 
higher-order TGCNs. Higher-order methods may
be occasionally better for some datasets.
Recall that the multi-layer first-order TGCN is a special case of single-layer higher-order polynomials as illustrated in Section \ref{model}.
The large number of parameters in the higher-order methods
may leads to highly likelihood of overfitting and local convergence,
thereby contributing to their less impressive outcomes.

\subsubsection{Depth and Polynomial Orders}
We also test the effects
of different polynomial orders and layer numbers. The accuracy and training time (200 epochs) for different polynomial orders (Type-3 as an example) are shown in the first group of plots in Fig. \ref{ori1}.
We note that performance improvement appears
to saturate beyond certain polynomial order.
Since the results also indicate  growing training time for higher order
polynomials, it would be more efficient to limit the polynomial order
to 2 or 3. We also test the performance of first-order 
TGCNs (Type-2 as an example) with different layers in the last two 
plots in Fig.\ref{ori1}, 
which also show that a 2-layer or 3-layer TGCN would typically suffice.

\subsubsection{Convergence}
We evaluate the convergence of different TGCN models in Fig. \ref{ori2}. Here, we report the accuracy of training data and validation data for the Cora dataset. 
Form the results, we can see that TGCN models can converge well in the citation network datasets.

\subsubsection{Training Efficiency}
\begin{table}[t]
	\centering
	\begin{threeparttable}[b]
		\centering
		\caption{Training Time per Epoch}
		\begin{tabular}{|c|c|c|c|c|c|}
			\hline
			Dataset& GCN&2L1KT1&2L1KT2&2L2KT3&2L2KT4\\\hline
			Cora&21.2ms&24.3ms&34.0ms&506.1ms&479.5ms\\\hline
			Citesser&30.5ms&33.4ms&42.9ms&681.1ms&726.5ms\\\hline
			Dataset&1L2KT3&1L3KT3&1L2KT4&1L3KT4 &1L1KT1  \\\hline
			Cora&253.8ms&463.2ms&244.5ms&641.0ms&11.3ms \\\hline
			Citesser&339.3ms&609.7ms&314.7ms&1085.6ms&33.2ms \\\hline
		\end{tabular}
		\label{ovtime}
		\begin{tablenotes} 
			\item [*] Different methods are measured in a CPU-only implementation.
			\item [*] $a$L$b$KT$c$ is short for a Type-$c$ TGCN with $a$ layers and polynomial order $k=b$.
		\end{tablenotes}
	\end{threeparttable}	
\end{table}
We compare the training efficiency for different methods based on the average training time for each epoch over $200$ epochs in total. 
We use the same number of hidden units for multi-layer graph convolutional networks to 
be fair. From the results of Table \ref{ovtime}, 2-layer TGCN takes nearly
$10\%$ longer than traditional 2-layer GCN because of
the larger number of parameters and matrix computations. 
Moreover, larger layer depth and higher polynomial order also increase
TGCN training time.

\begin{table}[t]
	\centering
	\caption{Performance of Different Graph Representations}
	\begin{tabular}{|l|l|l|l|}
		\hline
		Dataset& Cora & Citeseer & Pubmed \\ \hline\hline
		\multicolumn{4}{|l|}{Type-1 TGCN (2-Layer/Auto-training on $\alpha$)}       \\ \hline
		$\mathbf{A}$   &   76.0   &   67.8       &     77.9   \\ \hline 
		$\mathbf{D}^{-1}\mathbf{A}$&  79.3    &     \textbf{70.2}     &    \textbf{80.5}    \\ \hline
		$\tilde{\mathbf{D}}^{-\frac{1}{2}}\tilde{\mathbf{A}}\tilde{\mathbf{D}}^{-\frac{1}{2}}$&   78.6   &   69.2       &  78.9      \\ \hline
		$0.1\times(\mathbf{I}-0.9\tilde{\mathbf{A}})^{-1}$&  \textbf{80.1}    &    70.1      &    79.3    \\ \hline\hline
		\multicolumn{4}{|l|}{Type-2 TGCN (2-Layer)}       \\ \hline\hline
	    $\mathbf{A}$	&   76.8   &   67.6        & 77.8       \\ \hline
		$\mathbf{D}^{-1}\mathbf{A}$	&   \textbf{81.9}   &       70.0   &     \textbf{80.3}   \\ \hline
		$\tilde{\mathbf{D}}^{-\frac{1}{2}}\tilde{\mathbf{A}}\tilde{\mathbf{D}}^{-\frac{1}{2}}$&   81.6   &    70.1      &   79.4     \\ \hline
			$0.1\times(\mathbf{I}-0.9\tilde{\mathbf{A}})^{-1}$&    81.8  &    \textbf{70.1}      &    80.1    \\ \hline\hline
		\multicolumn{4}{|l|}{Type-3 TGCN (1-Layer $2^{nd}$-Order)}       \\ \hline\hline
		$\mathbf{A}$& 77.9     &68.5          &      /  \\ \hline
			$\mathbf{D}^{-1}\mathbf{A}$&  \textbf{79.0}    &      68.9    &       / \\ \hline
		$\tilde{\mathbf{D}}^{-\frac{1}{2}}\tilde{\mathbf{A}}\tilde{\mathbf{D}}^{-\frac{1}{2}}$&      78.5 & \textbf{69.1}         &    /    \\ \hline
			$0.1\times(\mathbf{I}-0.9\tilde{\mathbf{A}})^{-1}$& 78.8     &    68.7      &      /  \\ \hline\hline
		\multicolumn{4}{|l|}{Type-4 TGCN (1-Layer $2^{nd}$-Order)}       \\ \hline\hline
	    $\mathbf{A}$	&  75.2   &    67.0    &     /   \\ \hline
			$\mathbf{D}^{-1}\mathbf{A}$&  \textbf{79.7}     &    \textbf{68.2}      &      /  \\ \hline
		$\tilde{\mathbf{D}}^{-\frac{1}{2}}\tilde{\mathbf{A}}\tilde{\mathbf{D}}^{-\frac{1}{2}}$& 76.3      &     67.9       &     /   \\ \hline
			$0.1\times(\mathbf{I}-0.9\tilde{\mathbf{A}})^{-1}$&   78.3   &    68.0      &    /    \\ \hline
	\end{tabular}
\label{graphre}
\end{table}

\subsubsection{Different Choices of Graph Representations}
Thanks to the flexibility of $\mathbf{P}$ in the TGCN, it is 
interesting to explore different graph representations in different types of TGCN propagation models. Note that the Laplacian-based model
can be written in the form of the adjacency matrix and
a corresponding diagonal matrix, which can be included
within the category of adjacency-based convolutional
propagation models in TGCNs. Thus, we mainly investigate 
adjacency-based representation for TGCNs. 
Since there is no constraint on the range of eigenvalues, we test 
the effect of normalization. 
More specifically, we use the original adjacency matrix $\mathbf{A}$, the normalized adjacency matrix $\mathbf{D}^{-1}\mathbf{A}$, the traditional GCN propagation $\tilde{\mathbf{D}}^{-\frac{1}{2}}\tilde{\mathbf{A}}\tilde{\mathbf{D}}^{-\frac{1}{2}}$, and the pagerank $0.1\times(\mathbf{I}-0.9\tilde{\mathbf{A}})^{-1}$ \cite{klicpera2018predict}. For the higher-order TGCNs, 
we mainly focus on Cora and Citeseer datasets due to complexity.

The experiment results are shown in Table \ref{graphre}. The results
show that normalized representations exhibit better performance than the 
unnormalized graph representation for TGCN propagation. More specifically, the normalized adjacency matrix achieves better performance in most 
TGCN designs. Although the pagerank propagation also shows good performance in some of the TGCN categories, the high-complexity of calculating the inverse matrix is detrimental to its applications in 
higher-order TGCNs. Note that we only test some of the common graph representations of $\mathbf{P}$ in our TGCN designs. The Taylor expansions allow a more flexible combination with other existing GCN propagations, 
such as pagerank and GCN propagations. We plan to further investigate  alternative graph representations in future works.

\subsubsection{Discussion}
In terms of formulation, type-1 TGCN is an extension 
of GCN, which allows flexible self-influence for each node.
Type-2 TGCN is an extension of type-1 TGCN, where different self-influence parameters are 
assigned for different nodes. In practical applications, 
such self-influence does exist and may be less obvious.
Type-2 TGCN allows different self-influence parameters 
to be learned while training, which may lead to better performance 
in the citation networks.  Higher-order TGCNs, as discussed in Section \ref{model},
are
different from multi-layer TGCN as various orders
may lead to different performances. 
However, to mitigate complexity concerns, lower-order TGCNs 
are more efficient 
in applications. With the steady advances of 
computation hardware, 
higher-order TGCNs is expected to play increasingly 
important roles in future data analysis.
In addition, the TGCN designs show a scalable combination 
with existing GCN propagations and graph representations.

\subsection{Comparison with Several Existing Methods}\label{compare}
In this section, we compare our proposed TGCNs with several state-of-the-art methods in three different tasks: 1) node classification in the citation network; 2) point cloud segmentation; and 3) classification in synthetic datasets.
\begin{table}[t]
	\centering
	\caption{Comparison with Other Methods in the Citation Networks}
	\begin{tabular}{|l|l|l|l|}
		\hline
		& Cora           & Citeseer       & Pubmed         \\ \hline
		GCN         & 85.77          & 73.68          & 88.13          \\ \hline
		GAT         & 86.37          & 74.32          & 87.62          \\ \hline
		GIN         & 86.20          & 76.80          & 87.39          \\ \hline
		Geom-GCN-I  & 85.19          & 77.99          & \textbf{90.05} \\ \hline
		Geom-GCN-P  & 84.93          & 75.14          & 88.09          \\ \hline
		Geom-GCN-S  & 85.27          & 74.71          & 84.75          \\ \hline
		APPNP       & 86.88          & 77.74          & 88.41               \\ \hline
		Type-1 TGCN & 86.79          & 77.82          & 87.99          \\ \hline
		Type-2 TGCN & \textbf{87.23} & \textbf{78.31} & 86.89          \\ \hline
	\end{tabular}
\label{stoa}
\end{table}
\subsubsection{Classification in the Citation Networks}
We first compare the proposed TGCN frameworks with other GCN-style methods in the citation networks summarized in Table \ref{sta}, but with a different splits on the datasets. Instead of randomly splitting a subset of the citation networks, we apply similar data splits as \cite{pei2020geom} with 60\%/20\%/20\% for training/testing/validation datasets. We compare our methods with graph convolutional networks (GCN) \cite{kipf2017semi}, geometric graph convolutional networks (Geom-GCN) \cite{pei2020geom}, graph attention networks (GAT) \cite{velivckovic2017graph}, approximated personalized propagation of neural predictions (APPNP) \cite{klicpera2018predict}, and graph isomorphism networks (GIN) \cite{xu2018powerful}. In TGCN propagation, we use the normalized 
adjacency matrix, i.e., $\mathbf{P=D}^{-1}\mathbf{A}$, and set the number of layer to be two. The test results are reported in Table \ref{stoa}. The results show that the proposed TGCN achieves competitive performance 
against various GCN-type methods.  
Together with the  results presented in Section \ref{reu}, 
our experiments demonstrate that TGCN is very
effective in
node classification over the citation graphs despite data splits.
\begin{table*}[t]
	\centering
	\begin{threeparttable}[b]
		\centering
		\caption{Mean Accuracy in ShapeNet Dataset.}
		\begin{tabular}{|c|c|c|c|c|c|c|c|c|}
			\hline
			& Type-2 TGCN$^1$                            & Type-1 TGCN$^1$   & Type-2 TGCN$^2$    & Type-1 TGCN$^2$  & GSP & HGSP   & GCN          & PointNet
			   \\ \hline
			Airplane      & {{0.7551}}  & \textbf{0.7883}  &0.7785 & 0.7824 &0.5272
			 & 0.5566
			 & 0.7660       &   \textit{0.834}   \\ \hline
			Bag           & {0.9165}& {0.9202}  & 0.9399& \textbf{\textit{0.9409}} &0.5942 &0.5620 & 0.9176       &  0.787   \\ \hline
			Cap           & \textbf{0.7670}                         & 0.7599       & 0.7479&0.7577 & 0.6698&   0.7212  & 0.7629      &  \textit{0.825}   \\ \hline
			Car           & \textbf{0.7114}                        & 0.7052       & 0.7018& 0.6976 &0.3785& 0.3702    & 0.6790      &     \textit{0.749}  \\ \hline
			Chair         & {0.6603}                        & 0.6197       &\textbf{0.7412} &0.6885 & 0.4701&   0.5782  & 0.6430       &   \textit{0.896}   \\ \hline
			Earphone      & 0.7037                                 & {0.7135}& 0.7606 &\textbf{\textit{0.7712}} &0.4706 & 0.5637  & 0.7054      &    0.730  \\ \hline
			Guitar        & \textbf{0.8449}                        & 0.8401         & 0.8176&0.8265 & 0.5731&  0.5889 & 0.8304      &   \textit{0.915}   \\ \hline
			Knife         & \textbf{0.7675}                        & 0.7474         &0.7610 & 0.7614&0.6395 &  0.7045 & 0.7502        & \textit{0.859}   \\ \hline
			Lamp          & 0.7787                                 & {0.7836} & \textbf{0.7853}& 0.7712 &0.2510 & 0.3112 & 0.7821      &  \textit{0.808}    \\ \hline
			Laptop        & 0.8142                                 & {0.8365}  &0.8185 &0.8275 &0.6704 & \textbf{0.9077}& 0.8272       & \textit{0.953}   \\ \hline
			Motorbike     & 0.7167                                 & 0.7183          &0.7239 &0.7623 &\textit{\textbf{0.7663}} & 0.7588 & {0.7297}  &0.652 \\ \hline
			Mug           & 0.9324                                 & \textbf{\textit{0.9436}}  &  0.9348 & 0.9376 &0.7465 &0.6290 & 0.9302         & 0.930  \\ \hline
			Pistol        & 0.7362                                 & \textbf{0.7387}  &0.7145 &0.7107 &0.5336 &0.6277 & 0.7205        & \textit{0.812}   \\ \hline
			Rocket        & \textit{\textbf{0.7895}}                        & 0.7712          & 0.7824&0.7742 &0.4792 & 0.5481 & 0.7807        &  0.579  \\ \hline
			Skateboard    & {0.8323}                                 & 0.8364       &0.8230 &\textit{\textbf{0.8493}} &0.6088 &   0.5440  & {0.8176}  & 0.728\\ \hline
			Table         & 0.7984                                 & 0.8154         &0.7972 & 0.\textbf{\textit{8194}}&0.4726 & 0.4568  & {0.8164}         & 0.806  \\\midrule[1pt]
			\textbf{Mean} & 0.7828                             & {0.7836}& 0.7892&\textbf{0.7966} & 0.5532&0.5893& 0.7788    & \textit{0.803}    \\ \hline
		\end{tabular}	\label{accpts}
		\begin{tablenotes} 
			\item [1] TGCN$^1$ applies the graph representation $\mathbf{P}=\tilde{\mathbf{D}}^{-\frac{1}{2}}\tilde{\mathbf{A}}\tilde{\mathbf{D}}^{-\frac{1}{2}}$.
			\item [2] TGCN$^2$ applies the graph representation $\mathbf{P}=\mathbf{D}^{-1}\mathbf{A}$.
			\item [*] The best graph-based method is marked in bold font.
			\item [*] The best method is marked in italic script.
		\end{tablenotes}
	\end{threeparttable}
\end{table*}

\subsubsection{Point Cloud Segmentation}
We next test the performance of TGCN in the point cloud segmentation.
The goal of point cloud segmentation is to identify and cluster points in
a point cloud that share similar features into their respective regions \cite{nguyen20133d}. 
The segmentation problem can be posed as a semi-supervised classification 
problem if the labels of several samples are known \cite{lv2012semi}. 

\textbf{Datasets and Baselines}: In this work, we use the ShapeNet datasets \cite{chang2015shapenet,yi2017large} 
as examples. In this dataset, there are 16 object categories, each of which 
may contain 2-6 classes. We compare both type-1 and type-2 TGCNs 
with traditional GCN in all categories. To explore the features extracted from the graph convolution, we also compare the proposed methods with geometric clustering-based methods, including the graph spectral clustering (GSP) and hypergraph spectral clustering (HGSP) \cite{zhang2020hypergraph1}. A comparison with a specifically-designed neural networks for point cloud segmentation, i. e., PointNet \cite{qi2017pointnet}, is also reported.

\textbf{Experiment Setup}:
\begin{itemize}
	\item To implement TGCN efficiently, we randomly pick 
	20 point cloud objects from each category, and randomly set $70\%$ points as training data 
	with labels while using the remaining points
	as the test data for each point cloud.
	We use $k$-nearest neighbor method to construct an adjacency matrix $\mathbf{A}$ with 
	elements 
	$a_{ij}=1, \ 0$ to indicate the presence or the absence of connection between two nodes $i,j$,
	respectively. 
	More specifically, we set $k=20$ in graph construction for all point clouds. For the GCN-like methods, we fix the number of layer as two and the number hidden units as 40 to ensure a fair comparison. 
	\item For the spectral-clustering based methods, we use the hypergraph stationary process \cite{zhang2020hypergraph} to estimate the hypergraph spectrum for the HGSP-based method, and apply the Gaussian distance model \cite{8101025} to construct the graph for the GSP-based method. The $k$-means clustering is applied for segmentation after obtaining the key spectra.
\end{itemize}

\begin{figure*}[t]
	\centering
	\subfigure[Ground Truth.]{
		\centering
		\includegraphics[height=2.5in]{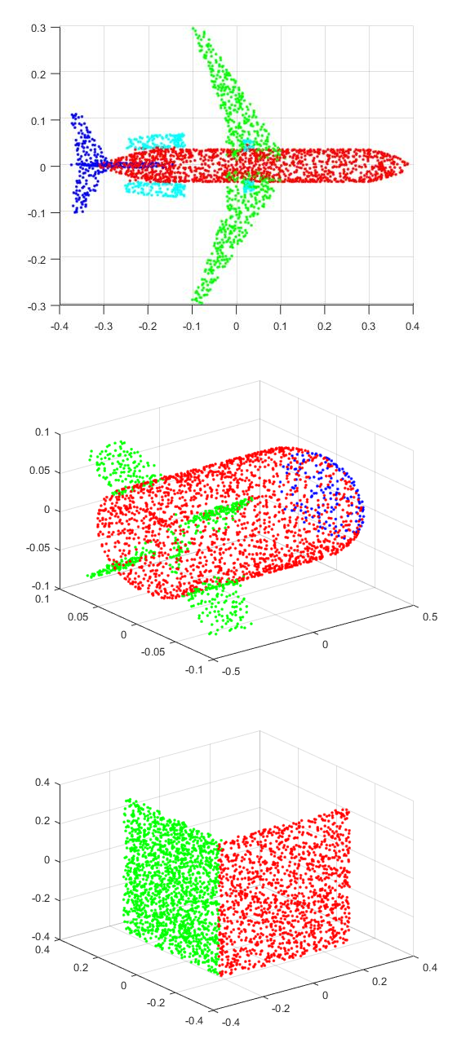}
	}
	\hfill
	\subfigure[TGCN.]{
		\centering
		\includegraphics[height=2.5in]{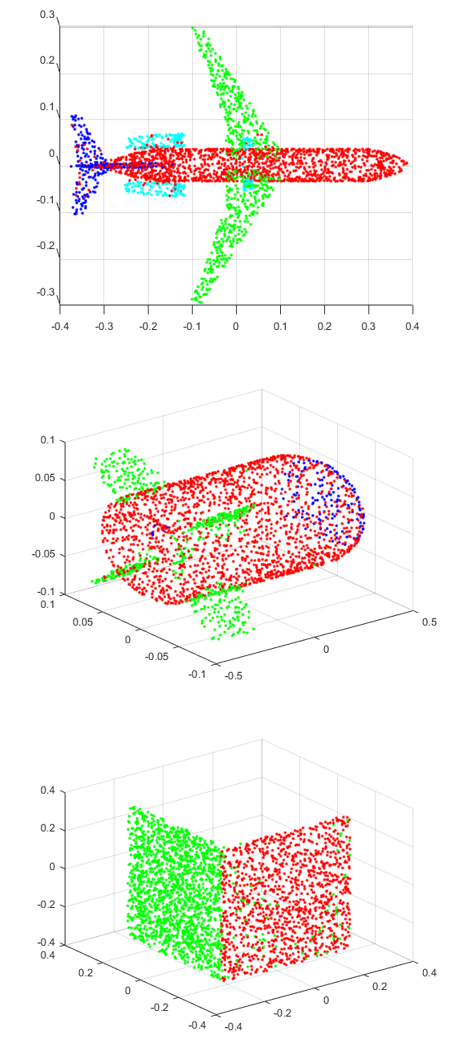}
	}
	\hfill
	\subfigure[GCN.]{
		\centering
		\includegraphics[height=2.5in]{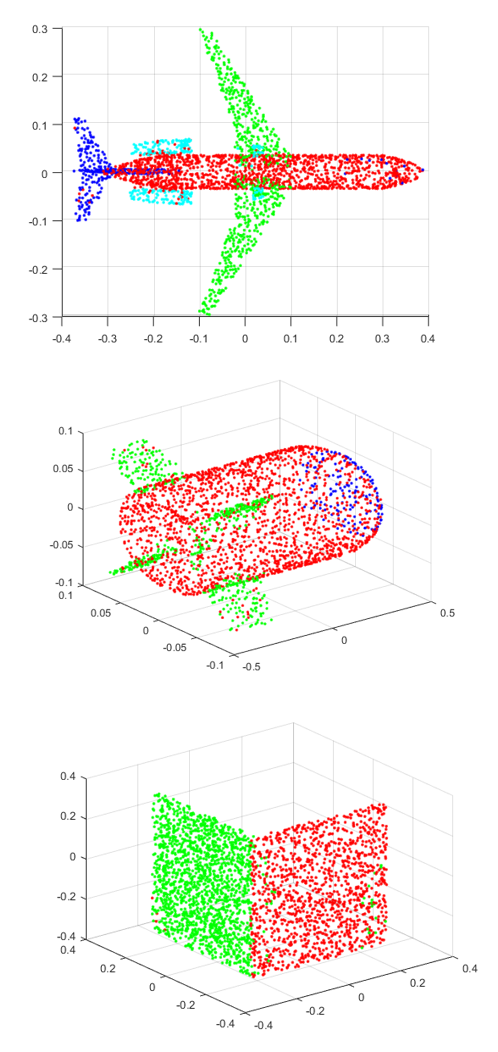}
	}
	\hfill
	\subfigure[GSP.]{
		\centering
		\includegraphics[height=2.5in]{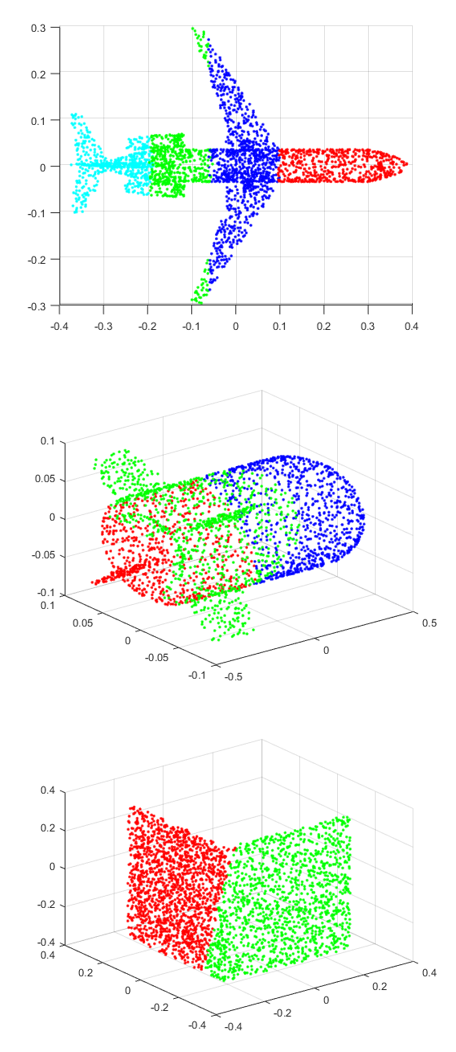}
	}
	\hfill
	\subfigure[HGSP.]{
	   \label{cub11}
		\centering
		\includegraphics[height=2.5in]{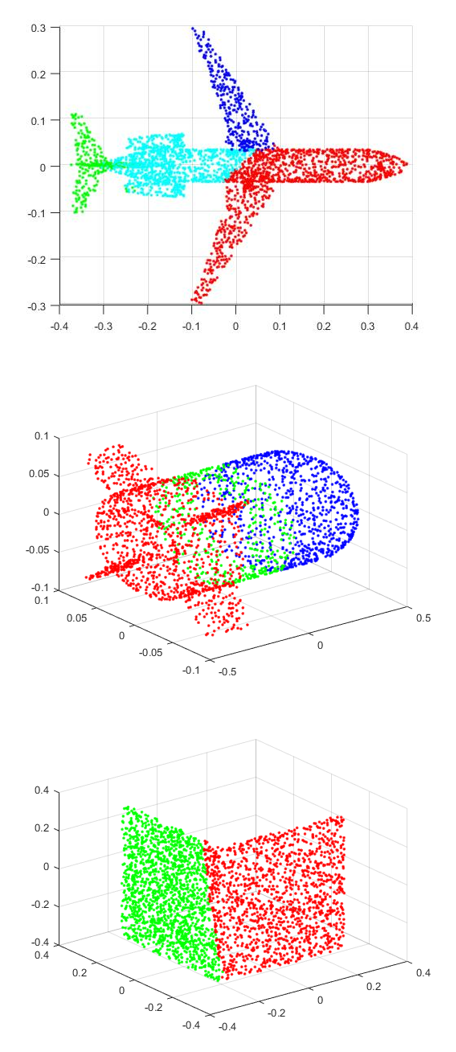}
	}
	\centering
	\caption{Examples of Point Cloud Segmentation.}
	\label{pts}
\end{figure*}

\textbf{Experiment Results}:
The overall accuracies of different methods are reported in Table 
\ref{accpts}.
 The resulting mean accuracy 
of segmentation illustrates that each method under comparison may
exhibit some unique strength in different categories. 
The PointNet exhibits an overall largest mean accuracy, since it is specifically designed for point cloud segmentation while the GCN-like methods are directly applied in classifying the points without adjustment. Even though, the TGCN still shows better performances in some of the categories than PointNet, such as table, skateboard, mug, motorbike, earphone and bag. Compared to the traditional GCN,
TGCN generally achieves higher accuracy 
and provides the better performance in most of the categories. For the graph representation, the normalized adjacency matrix performs better than the traditional GCN propagation matrix in the TGCN designs for point cloud segmentation.
The clustering-based methods perform worse than the classification-based methods, since they use no prior knowledge of the ground truth. However, they still achieve satisfying performances in the point clouds with fewer classes and more regular shapes.

 To further illustrate different methods, several visualized segmentation results are presented in Fig. \ref{pts}. Since different TGCN exhibits similar visualized results, we report type-2 TGCN with $\mathbf{P}=\mathbf{D}^{-1}\mathbf{A}$ as an example. From the results, we see that the classification-based methods, i. e., TGCN and GCN, exhibit similar results, where errors are distributed 
 scatteredly over the point clouds. Generally, the TGCN results show fewer errors than those of the traditional GCN, such as in the wings of rockets and planes. Different from GCN-like methods trained according to the ground truth, the clustering-based methods show different results. For example, in the first row of Fig. \ref{cub11}, although the segmentation result differ from the ground
 truth, these results still make sense by grouping two wings to
 different classes. In addition, the errors of clustering-based methods are grouped together in the intersections of two classes. It will be interesting to explore the integration of classification-based methods and clustering-based methods to extract more features of point clouds in the future works.

\begin{table}[t]
	\centering
	\caption{Tests in the Synthetic Datasets}
	\begin{tabular}{|l|l|l|l|l|}
		\hline
		& \multicolumn{2}{l|}{Pattern} & \multicolumn{2}{l|}{Cluster} \\ \hline
		Method      & Layer       & Accuracy       & Layer       & Accuracy       \\ \hline
		GCN         & 4           & 63.88          & 4           & 53.4         \\ \hline
		GraphSage   & 4           & 50.52          & 4           & 68.5         \\ \hline
		GAT         & 4           & 75.84           & 4           & 58.3           \\ \hline
		GateGCN     & 4           & 84.53           & 4           & 60.4           \\ \hline
		GIN         & 4           & 85.59          & 4           & 58.3           \\ \hline
		RingGNN     & 2           & 86.24          & 2           & 42.4           \\ \hline
		Type-1 TGCN & 4           & 85.43          & 4           & 58.8           \\ \hline
		Type-2 TGCN & 4           & 85.87          & 4           & 61.2           \\ \hline
	\end{tabular}
\label{cc}
\end{table}

\subsubsection{Node Classification in Synthetic Datasets}
To provide more comprehensive results of the proposed TGCNs, we also test on the node classification on synthetic datasets following the stochastic block models (SBM), which are widely used to model communities in social networks. We randomly generated SBM graphs with a smaller data size following the similar strategies as \cite{dwivedi2020benchmarking} with 70\%/15\%/15\% for the training/test/validation. In designing TGCNs, we set the graph representation $\mathbf{P}=\mathbf{D}^{-1}\mathbf{A}$. An overall result is reported in Table \ref{cc}. Although the proposed TGCNs do not deliver
the top performance in each test case, they are generally
robust and competitive across the board.

\section{Conclusion}
In this work, we explore the inherent connection 
between GSP convolutional spectrum 
wavelet and the GCN vertex propagation. 
Our work shows that 
spectral wavelet-kernel can be approximated in vertex domain 
if it admits a polynomial approximation. 
In addition, we develop an efficient and simple
alternative design of GCN layers
based on the simple Taylor expansion (TGCN), 
which exhibits computation efficiency and outperforms 
a number of state-of-the-art GCN-type methods. 
Our work derives practial guidelines on the selection of
representing matrix and  the propagation model for TGCN designs. 
Our interpretability discussion presents good insights into 
Taylor approximation of graph convolution. 

\textbf{Future Works}: In existing works, the design of GCN 
centers on performance while neglecting the 
underlying connection to the original 
graph spectrum convolution in GSP. 
Evaluating GCN from the GSP-perspective may provide 
better insights for layer design and basis for performance 
analysis in the future. 
It is equally important for future works
to explore the choice of representing matrix and propagation models to 
further enhance the performance and robustness of GCNs.

\bibliographystyle{IEEEtran}
\bibliography{list}

\begin{thebibliography}{10}
\providecommand{\url}[1]{#1}
\csname url@samestyle\endcsname
\providecommand{\newblock}{\relax}
\providecommand{\bibinfo}[2]{#2}
\providecommand{\BIBentrySTDinterwordspacing}{\spaceskip=0pt\relax}
\providecommand{\BIBentryALTinterwordstretchfactor}{4}
\providecommand{\BIBentryALTinterwordspacing}{\spaceskip=\fontdimen2\font plus
\BIBentryALTinterwordstretchfactor\fontdimen3\font minus
  \fontdimen4\font\relax}
\providecommand{\BIBforeignlanguage}[2]{{%
\expandafter\ifx\csname l@#1\endcsname\relax
\typeout{** WARNING: IEEEtran.bst: No hyphenation pattern has been}%
\typeout{** loaded for the language `#1'. Using the pattern for}%
\typeout{** the default language instead.}%
\else
\language=\csname l@#1\endcsname
\fi
#2}}
\providecommand{\BIBdecl}{\relax}
\BIBdecl

\bibitem{kipf2017semi}
T.~N. Kipf and M.~Welling, ``Semi-supervised classification with graph
  convolutional networks,'' in \emph{International Conference on Learning
  Representations (ICLR)}, Toulon, France, Apr. 2017.

\bibitem{zhu2003semi}
X.~Zhu, Z.~Ghahramani, and J.~D. Lafferty, ``Semi-supervised learning using
  gaussian fields and harmonic functions,'' in \emph{Proceedings of the 20th
  International Conference on Machine learning (ICML-03)}, Washington DC, USA,
  Aug. 2003, pp. 912--919.

\bibitem{white2005spectral}
S.~White and P.~Smyth, ``A spectral clustering approach to finding communities
  in graphs,'' in \emph{Proceedings of the 2005 SIAM International Conference
  on Data Mining}, Newport Beach, CA, USA, 2005, pp. 274--285.

\bibitem{buhler2009spectral}
T.~B{\"u}hler and M.~Hein, ``Spectral clustering based on the graph
  p-laplacian,'' in \emph{Proceedings of the 26th Annual International
  Conference on Machine Learning}, New York, NY, USA, Jun. 2009, pp. 81--88.

\bibitem{grover2016node2vec}
A.~Grover and J.~Leskovec, ``node2vec: Scalable feature learning for
  networks,'' in \emph{Proceedings of the 22nd ACM SIGKDD International
  Conference on Knowledge Discovery and Data Mining}, San Francisco, CA, USA,
  2016, pp. 855--864.

\bibitem{niepert2016learning}
M.~Niepert, M.~Ahmed, and K.~Kutzkov, ``Learning convolutional neural networks
  for graphs,'' in \emph{International Conference on Machine Learning}, New
  York, NY, USA, Jun. 2016, pp. 2014--2023.

\bibitem{8101025}
S.~{Chen}, D.~{Tian}, C.~{Feng}, A.~{Vetro}, and J.~{Kovačević}, ``Fast
  resampling of three-dimensional point clouds via graphs,'' \emph{IEEE
  Transactions on Signal Processing}, vol.~66, no.~3, pp. 666--681, 2018.

\bibitem{ortega2018graph}
A.~Ortega, P.~Frossard, J.~Kova{\v{c}}evi{\'c}, J.~M. Moura, and
  P.~Vandergheynst, ``Graph signal processing: Overview, challenges, and
  applications,'' \emph{Proceedings of the IEEE}, vol. 106, no.~5, pp.
  808--828, May 2018.

\bibitem{shuman2013emerging}
D.~I. Shuman, S.~K. Narang, P.~Frossard, A.~Ortega, and P.~Vandergheynst, ``The
  emerging field of signal processing on graphs: Extending high-dimensional
  data analysis to networks and other irregular domains,'' \emph{IEEE Signal
  Processing Magazine}, vol.~30, no.~3, pp. 83--98, Apr. 2013.

\bibitem{chen2014semi}
S.~Chen, F.~Cerda, P.~Rizzo, J.~Bielak, J.~H. Garrett, and
  J.~Kova{\v{c}}evi{\'c}, ``Semi-supervised multiresolution classification
  using adaptive graph filtering with application to indirect bridge structural
  health monitoring,'' \emph{IEEE Transactions on Signal Processing}, vol.~62,
  no.~11, pp. 2879--2893, Jun. 2014.

\bibitem{schoenenberger2015graph}
Y.~Schoenenberger, J.~Paratte, and P.~Vandergheynst, ``Graph-based denoising
  for time-varying point clouds,'' in \emph{2015 3DTV-Conference: The True
  Vision-Capture, Transmission and Display of 3D Video (3DTV-CON)}, Lisbon,
  Portugal, Jul. 2015, pp. 1--4.

\bibitem{sandryhaila2013discrete1}
A.~Sandryhaila and J.~M. Moura, ``Discrete signal processing on graphs: Graph
  filters,'' in \emph{2013 IEEE International Conference on Acoustics, Speech
  and Signal Processing}, Vancouver, Canada, May 2013, pp. 6163--6166.

\bibitem{sandryhaila2013discrete}
------, ``Discrete signal processing on graphs: Graph fourier transform,'' in
  \emph{2013 IEEE International Conference on Acoustics, Speech and Signal
  Processing}, Vancouver, Canada, May 2013, pp. 6167--6170.

\bibitem{hammond2011wavelets}
D.~K. Hammond, P.~Vandergheynst, and R.~Gribonval, ``Wavelets on graphs via
  spectral graph theory,'' \emph{Applied and Computational Harmonic Analysis},
  vol.~30, no.~2, pp. 129--150, Mar. 2011.

\bibitem{shuman2012windowed}
D.~I. Shuman, B.~Ricaud, and P.~Vandergheynst, ``A windowed graph fourier
  transform,'' in \emph{2012 IEEE Statistical Signal Processing Workshop
  (SSP)}, Ann Arbor, USA, Aug. 2012, pp. 133--136.

\bibitem{zhang2020hypergraph}
S.~{Zhang}, S.~{Cui}, and Z.~{Ding}, ``Hypergraph-based image processing,'' in
  \emph{2020 IEEE International Conference on Image Processing (ICIP)}, Abu
  Dhabi, United Arab Emirates, 2020, pp. 216--220.

\bibitem{klicpera2018predict}
J.~Klicpera, A.~Bojchevski, and S.~G{\"u}nnemann, ``Predict then propagate:
  Graph neural networks meet personalized pagerank,'' \emph{arXiv preprint
  arXiv:1810.05997}, 2018.

\bibitem{abu2018n}
S.~Abu-El-Haija, A.~Kapoor, B.~Perozzi, and J.~Lee, ``N-gcn: Multi-scale graph
  convolution for semi-supervised node classification,'' \emph{arXiv preprint
  arXiv:1802.08888}, 2018.

\bibitem{wagner2006distributed}
R.~Wagner, V.~Delouille, and R.~Baraniuk, ``Distributed wavelet de-noising for
  sensor networks,'' in \emph{Proceedings of the 45th IEEE Conference on
  Decision and Control}, San Diego, CA, USA, Dec. 2006, pp. 373--379.

\bibitem{narang2010local}
S.~K. Narang and A.~Ortega, ``Local two-channel critically sampled filter-banks
  on graphs,'' in \emph{2010 IEEE International Conference on Image
  Processing}, Hong Kong, China, Jan. 2010, pp. 333--336.

\bibitem{zhu2012approximating}
X.~Zhu and M.~Rabbat, ``Approximating signals supported on graphs,'' in
  \emph{2012 IEEE International Conference on Acoustics, Speech and Signal
  Processing (ICASSP)}, Japan, Mar. 2012, pp. 3921--3924.

\bibitem{chen2015discrete}
S.~Chen, R.~Varma, A.~Sandryhaila, and J.~Kova{\v{c}}evi{\'c}, ``Discrete
  signal processing on graphs: sampling theory,'' \emph{IEEE Transactions on
  Signal Processing}, vol.~63, no.~24, pp. 6510--6523, Dec. 2015.

\bibitem{sandryhaila2014discrete}
A.~Sandryhaila and J.~M. Moura, ``Discrete signal processing on graphs:
  Frequency analysis,'' \emph{IEEE Transactions on Signal Processing}, vol.~62,
  no.~12, pp. 3042--3054, Apr. 2014.

\bibitem{marques2017stationary}
A.~G. Marques, S.~Segarra, G.~Leus, and A.~Ribeiro, ``Stationary graph
  processes and spectral estimation,'' \emph{IEEE Transactions on Signal
  Processing}, vol.~65, no.~22, pp. 5911--5926, Aug. 2017.

\bibitem{zhang2019introducing}
S.~Zhang, Z.~Ding, and S.~Cui, ``Introducing hypergraph signal processing:
  Theoretical foundation and practical applications,'' \emph{IEEE Internet of
  Things Journal}, vol.~7, no.~1, Jan. 2020.

\bibitem{barbarossa2020topological}
S.~Barbarossa and S.~Sardellitti, ``Topological signal processing over
  simplicial complexes,'' \emph{IEEE Transactions on Signal Processing}, Mar.
  2020.

\bibitem{page1999pagerank}
L.~Page, S.~Brin, R.~Motwani, and T.~Winograd, ``The pagerank citation ranking:
  Bringing order to the web.'' Stanford InfoLab, Tech. Rep., 1999.

\bibitem{bresson2017residual}
X.~Bresson and T.~Laurent, ``Residual gated graph convnets,'' \emph{arXiv
  preprint arXiv:1711.07553}, 2017.

\bibitem{hamilton2017inductive}
W.~Hamilton, Z.~Ying, and J.~Leskovec, ``Inductive representation learning on
  large graphs,'' in \emph{Advances in Neural Information Processing Systems},
  Long Beach, USA, Dec. 2017, pp. 1024--1034.

\bibitem{monti2017geometric}
F.~Monti, D.~Boscaini, J.~Masci, E.~Rodola, J.~Svoboda, and M.~M. Bronstein,
  ``Geometric deep learning on graphs and manifolds using mixture model cnns,''
  in \emph{Proceedings of the IEEE Conference on Computer Vision and Pattern
  Recognition}, Hawaii, USA, Jul. 2017, pp. 5115--5124.

\bibitem{velivckovic2017graph}
P.~Veli{\v{c}}kovi{\'c}, G.~Cucurull, A.~Casanova, A.~Romero, P.~Lio, and
  Y.~Bengio, ``Graph attention networks,'' \emph{arXiv preprint
  arXiv:1710.10903}, 2017.

\bibitem{ying2018hierarchical}
Z.~Ying, J.~You, C.~Morris, X.~Ren, W.~Hamilton, and J.~Leskovec,
  ``Hierarchical graph representation learning with differentiable pooling,''
  in \emph{Advances in Neural Information Processing Systems}, Montreal,
  Canada, Dec. 2018, pp. 4800--4810.

\bibitem{pei2020geom}
H.~Pei, B.~Wei, K.~C.-C. Chang, Y.~Lei, and B.~Yang, ``Geom-gcn: Geometric
  graph convolutional networks,'' \emph{arXiv preprint arXiv:2002.05287}, 2020.

\bibitem{abu2019mixhop}
S.~Abu-El-Haija, B.~Perozzi, A.~Kapoor, N.~Alipourfard, K.~Lerman,
  H.~Harutyunyan, G.~V. Steeg, and A.~Galstyan, ``Mixhop: Higher-order graph
  convolutional architectures via sparsified neighborhood mixing,'' \emph{arXiv
  preprint arXiv:1905.00067}, 2019.

\bibitem{atwood2016diffusion}
J.~Atwood and D.~Towsley, ``Diffusion-convolutional neural networks,'' in
  \emph{Advances in Neural Information Processing Systems}, 2016, pp.
  1993--2001.

\bibitem{xu2018powerful}
K.~Xu, W.~Hu, J.~Leskovec, and S.~Jegelka, ``How powerful are graph neural
  networks?'' \emph{arXiv preprint arXiv:1810.00826}, 2018.

\bibitem{bronstein2017geometric}
M.~M. Bronstein, J.~Bruna, Y.~LeCun, A.~Szlam, and P.~Vandergheynst,
  ``Geometric deep learning: going beyond euclidean data,'' \emph{IEEE Signal
  Processing Magazine}, vol.~34, no.~4, pp. 18--42, Jul. 2017.

\bibitem{dwivedi2020benchmarking}
V.~P. Dwivedi, C.~K. Joshi, T.~Laurent, Y.~Bengio, and X.~Bresson,
  ``Benchmarking graph neural networks,'' \emph{arXiv preprint
  arXiv:2003.00982}, 2020.

\bibitem{shi2019graph}
J.~Shi and J.~M. Moura, ``Graph signal processing: Modulation, convolution, and
  sampling,'' \emph{arXiv preprint arXiv:1912.06762}, 2019.

\bibitem{pons1998legendre}
J.~Pons, J.~Miralles, and J.~M. Ib{\'a}{\~n}ez, ``Legendre expansion of the
  kernel: Influence of high order terms,'' \emph{Astronomy and Astrophysics
  Supplement Series}, vol. 129, no.~2, pp. 343--351, 1998.

\bibitem{linnainmaa1976taylor}
S.~Linnainmaa, ``Taylor expansion of the accumulated rounding error,''
  \emph{BIT Numerical Mathematics}, vol.~16, no.~2, pp. 146--160, 1976.

\bibitem{mccallum2000automating}
A.~K. McCallum, K.~Nigam, J.~Rennie, and K.~Seymore, ``Automating the
  construction of internet portals with machine learning,'' \emph{Information
  Retrieval}, vol.~3, no.~2, pp. 127--163, Jul. 2000.

\bibitem{bojchevski2017deep}
A.~Bojchevski and S.~G{\"u}nnemann, ``Deep gaussian embedding of graphs:
  Unsupervised inductive learning via ranking,'' \emph{arXiv preprint
  arXiv:1707.03815}, 2017.

\bibitem{sen2008collective}
P.~Sen, G.~Namata, M.~Bilgic, L.~Getoor, B.~Galligher, and T.~Eliassi-Rad,
  ``Collective classification in network data,'' \emph{AI magazine}, vol.~29,
  no.~3, pp. 93--93, 2008.

\bibitem{namata2012query}
G.~Namata, B.~London, L.~Getoor, B.~Huang, and U.~EDU, ``Query-driven active
  surveying for collective classification,'' in \emph{10th International
  Workshop on Mining and Learning with Graphs}, vol.~8, Edinburgh, Scotland,
  Jul. 2012.

\bibitem{kingma2014adam}
D.~P. Kingma and J.~Ba, ``Adam: A method for stochastic optimization,''
  \emph{arXiv preprint arXiv:1412.6980}, 2014.

\bibitem{nguyen20133d}
A.~Nguyen and B.~Le, ``3d point cloud segmentation: A survey,'' in \emph{2013
  6th IEEE Conference on Robotics, Automation and Mechatronics (RAM)}, Manila,
  Philippines, Nov. 2013, pp. 225--230.

\bibitem{lv2012semi}
J.~Lv, X.~Chen, J.~Huang, and H.~Bao, ``Semi-supervised mesh segmentation and
  labeling,'' in \emph{Computer Graphics Forum}, vol.~31, no.~7.\hskip 1em plus
  0.5em minus 0.4em\relax Wiley Online Library, 2012, pp. 2241--2248.

\bibitem{chang2015shapenet}
A.~X. Chang, T.~Funkhouser, L.~Guibas, P.~Hanrahan, Q.~Huang, Z.~Li,
  S.~Savarese, M.~Savva, S.~Song, H.~Su \emph{et~al.}, ``Shapenet: An
  information-rich 3d model repository,'' \emph{arXiv preprint
  arXiv:1512.03012}, 2015.

\bibitem{yi2017large}
L.~Yi, L.~Shao, M.~Savva, H.~Huang, Y.~Zhou, Q.~Wang, B.~Graham, M.~Engelcke,
  R.~Klokov, V.~Lempitsky \emph{et~al.}, ``Large-scale 3d shape reconstruction
  and segmentation from shapenet core55,'' \emph{arXiv preprint
  arXiv:1710.06104}, 2017.

\bibitem{zhang2020hypergraph1}
S.~Zhang, S.~Cui, and Z.~Ding, ``Hypergraph spectral clustering for point cloud
  segmentation,'' \emph{IEEE Signal Processing Letters}, vol.~27, pp.
  1655--1659, 2020.

\bibitem{qi2017pointnet}
C.~R. Qi, H.~Su, K.~Mo, and L.~J. Guibas, ``Pointnet: Deep learning on point
  sets for 3d classification and segmentation,'' in \emph{Proceedings of the
  IEEE Conference on Computer Vision and Pattern Recognition}, 2017, pp.
  652--660.

\end{thebibliography}

\end{document}